\documentclass{fairmeta}

\usepackage{amsmath, amssymb, amsthm}
\usepackage{multirow}
\usepackage{tikz}
\usetikzlibrary{shapes.geometric}
\definecolor{metablue}{HTML}{0064E0}
\definecolor{metaorange}{HTML}{DC6414}

\usepackage{fancyhdr}
\ifdefined\arxivbuild
  \newcommand{\gitstamp}{}
\else
  \IfFileExists{gitrev.tex}{\newcommand{\gitstamp}{}
}{}
\fi
\providecommand{\gitstamp}{version unknown}
\fancypagestyle{plain}{%
  \fancyhf{}%
  \fancyfoot[C]{\thepage}%
  \fancyfoot[R]{\footnotesize\textcolor{gray}{\gitstamp}}%
}

\newcommand{\ATE}{\mathrm{ATE}}
\newcommand{\Binom}{\mathrm{Binom}}
\newcommand{\mpz}{m\sb{+0}}
\newcommand{\mpo}{m\sb{+1}}
\newcommand{\mpzo}{m\sb{+0}^0}
\newcommand{\mpoo}{m\sb{+1}^0}
\newcommand{\Szz}{S\sb{00}}
\newcommand{\Szo}{S\sb{01}}
\newcommand{\Soz}{S\sb{10}}
\newcommand{\Soo}{S\sb{11}}
\newcommand{\Sjk}{S\sb{jk}}
\newcommand{\mzz}{m\sb{00}}
\newcommand{\mzo}{m\sb{01}}
\newcommand{\moz}{m\sb{10}}
\newcommand{\moo}{m\sb{11}}
\newcommand{\mjk}{m\sb{jk}}
\newcommand{\az}{a\sb 0}
\newcommand{\Az}{A\sb 0}
\newcommand{\Ao}{A\sb 1}
\newcommand{\Hz}{H\sb 0}
\newcommand{\deltaz}{\delta\sb 0}
\newcommand{\Deltaz}{\Delta\sb 0}
\newcommand{\Ho}{H\sb 1}
\newcommand{\pgt}{p\sb{>}}
\newcommand{\plt}{p\sb{<}}
\newcommand{\zalpha}{z\sb{1-\alpha}}
\newcommand{\zalphaq}{z\sb{1-\alpha/4}}
\newcommand{\Lalpha}{L\sb \alpha}
\newcommand{\Ualpha}{U\sb \alpha}
\newcommand{\mop}{m\sb{1+}}
\newcommand{\mzp}{m\sb{0+}}
\newcommand{\znf}{z\sb{0.975}}
\newcommand{\zalphah}{z\sb{1-\alpha/2}}
\newcommand{\pis}{\pi\sb s}
\newcommand{\Sbullet}{S\sb\bullet}

\newtheorem{proposition}{Proposition}
\newtheorem{lemma}{Lemma}

\title{Randomization Inference for Matched Pairs with Binary Outcomes}

\author[1,*]{Bob Wilson}
\affiliation[1]{Marketing Data Science at Meta}
\contribution[*]{Work done at Meta}
\date{August 2026}
\correspondence{Bob Wilson at \email{rwilson4@meta.com}}

\abstract{We give an exact randomization-based confidence set for the
  average treatment effect ($\ATE$) in matched-pair studies with a
  binary outcome, requiring neither monotonicity nor any
  distributional assumption beyond the within-pair coin flip. At its
  core is an \emph{analytic} solution to the worst-case allocation of
  attributable effects: two binomial-symmetry lemmas identify the
  pattern hardest to reject as a single boundary corner, so testing
  null hypotheses needs no integer program and no numerical search.
  Inverting the test via binary search yields a prediction set for the
  attributable effect in $O(\log S)$ Binomial tail calculations; the
  Bonferroni proposition of \citet{rigdon-2015-attributable-effect}
  produces the $\ATE$ confidence set at the same computational cost. The same corner
  extends without further machinery to a sensitivity analysis for
  matched observational studies under Rosenbaum's $\Gamma$-model. A
  simple formula for the \emph{design sensitivity} illuminates when an
  observational study can hope to provide evidence for an effect.}

\begin{document}
\maketitle

\section{Introduction}
Matched-pair designs are used for both experimental and observational
causal inference. Units may be paired as a variance reduction
strategy, with one unit from each pair randomly selected for
treatment. In an observational study, treated units may be matched to
controls on the basis of observed covariates to reduce selection bias.
Inference about a treatment effect in such a design ought to rest on
no more than the structure imposed by the assignment mechanism,
whether imposed by the experiment or assumed in the matching.

When the hypothesis of constant treatment effects is reasonable, we
may use the Hodges-Lehmann method to obtain a point estimate
\citep{hodges-1963-point-estimates, rosenbaum-1993-hodges-lehmann-point-estimates}.
For a binary outcome that hypothesis is implausible, since a
constant nonzero effect forces every untreated unit to a single common counterfactual value,
incompatible with anything but a perfectly homogeneous population.
\citet{rosenbaum-2001-effects-attributable-to-treatment} introduced
the notion of an \emph{attributable effect} as an effect size
identified by randomization. It is the number of successes caused by
treatment among the treated, $\Ao$, and its mirror $\Az$ for the
untreated, the number of control failures that would have been
successes under treatment.

\citet{rosenbaum-2001-effects-attributable-to-treatment} built exact
prediction sets for $\Ao$ in the unstratified case by inverting
Fisher's exact test on a discrete pivot, under the \emph{monotonicity
assumption} that treatment never prevents a success.
\citet{rosenbaum-2002-attributing-effects-to-treatment-in-matched-studies}
extended attributable effects to matched observational studies, giving
exact randomization inference together with a large-sample
approximation via asymptotic separability, both under monotonicity.
\citet{rigdon-2015-attributable-effect} dropped monotonicity in the
unstratified randomized experiment and combined $\Ao$ with $\Az$ via
Bonferroni into a confidence set for the average treatment effect;
their procedure is fast because each prediction set inverts a single
hypergeometric. They consider stratified problems (including matched
pairs) only briefly, but propose a generic integer programming
solution impractical in large sample sizes.

The computational difficulty of exact randomization inference for
binary outcomes has drawn recent attention.
\citet{aronow-2025-fast-computation} reduce the computational cost of the
\emph{unstratified} interval of \citet{rigdon-2015-attributable-effect}
to $O(n \log n)$ permutation tests in balanced designs, but leave the
stratified case explicitly open. \citet{li-2025-exact-stratified}
treat the stratified binary $\ATE$ directly, obtaining exact and
conservative intervals by maximizing a permutation $p$-value over the
compatible potential-outcome tables through a pruned search, at computational cost
polynomial in the stratum sizes and with no closed form for the worst
case. \citet{zhang-2026-fast-algorithms} gives an $O(\log n)$
matched-pair $\ATE$ interval, also locating the worst case at a
boundary corner; there, however, the corner is found for each pair
stratum by solving a fixed-dimension integer program with a generic
solver, and the construction targets the $\ATE$ alone. None of these
address hidden bias.

This note gives an \emph{analytic} solution to the worst-case
allocation for matched pairs, and builds the full attributable-effect
and sensitivity machinery around it. Our contributions are:
\begin{enumerate}
  \item \textbf{A closed-form worst-case pattern.} Two
    binomial-symmetry lemmas (\S\ref{sec:hardest}) identify the
    pattern of treatment effects hardest to reject among the
    combinatorial number of patterns consistent with a hypothesized
    $\Ao = \az$, as the unique boundary corner of a two-dimensional
    integer rectangle. The worst-case test is a single Binomial tail
    calculation, with no integer program and no numerical search.
  \item \textbf{A point estimate for the $\ATE$} matching the textbook
    McNemar formula, but motivated by random assignment rather than
    sampling arguments.
  \item \textbf{An $\,O(\log S) \,$ $\ATE$ confidence set without
    monotonicity.} Sweeping $\az$ inverts the test in $O(\log S)$
    Binomial tail calculations; the mirror procedure handles $\Az$,
    and the Bonferroni combination of
    \citet{rigdon-2015-attributable-effect} delivers the $\ATE$ set at
    the same computational cost, assuming no more than the within-pair coin flip. A
    continuity-corrected Gaussian approximation
    (\S\ref{sec:practical}) delivers the endpoints and the $\ATE$
    interval in closed form.
  \item \textbf{A sensitivity analysis.} The construction extends to
    matched observational studies via Rosenbaum's $\Gamma$-model at no
    additional computational cost (\S\ref{sec:sensitivity}), including
    adjusted p-values, closed-form bounds on the Hodges-Lehmann point
    estimate, and expanded confidence intervals.
  \item \textbf{A simple formula for the design sensitivity of a
    matched pairs study.} The design sensitivity tells us the
    threshold of ``hidden biases'' beyond which the null hypothesis of
    no net treatment effect cannot be rejected at any sample size,
    playing a similar role in observational studies that power
    analysis plays in experiment design.
\end{enumerate}

\S\ref{sec:setting} sets up the matched-pair design, the random
attributable effects, and the prediction-set framing.
\S\ref{sec:nulls} reduces the testing problem to a family of
equivalence classes via McNemar's test. \S\ref{sec:hardest} proves two
monotonicity lemmas and identifies the single hardest pattern.
\S\ref{sec:practical} inverts the test by binary search, provides a
large-sample approximation, and proves the standard point estimate may
be motivated by random assignment alone. \S\ref{sec:ate} combines
$\Ao$ and $\Az$ into the $\ATE$ confidence set and compares with
sampling-based sets. \S\ref{sec:sensitivity} extends the procedure to
observational studies via Rosenbaum's sensitivity model.
\S\ref{sec:discussion} summarizes the main takeaways.

\section{Setting, target, and the randomization test}\label{sec:setting}
Suppose we have $2S$ units arranged in $S$ pairs, where units within a
pair are matched on covariates measured before the study. In an
experiment, one unit is selected for treatment uniformly at random
within each pair. In an observational study, each treated unit is paired with a
control unit, and we assume for now that within-pair selection
probabilities are equal and strictly between 0 and 1 (\emph{strongly
ignorable treatment selection}); \S\ref{sec:sensitivity} relaxes this.

Let $Z_{si} = 1$ if unit $i$ in pair $s$ is treated and $0$ otherwise,
with $\sum_{i=1}^2 Z_{si} = 1$. We assume the outcome experienced by a
unit depends on that unit's treatment selection, but not the selection
made by any other unit, the Stable Unit Treatment Value Assumption \citep{rubin-1980-sutva}. As a result,
each unit has potential outcomes $\{r_{Csi}, r_{Tsi}\} \in \{0, 1\}$,
and the observed outcome is $R_{si} = Z_{si} r_{Tsi} + (1 - Z_{si}) r_{Csi}$.
Let $R_{si}' = (1 - Z_{si}) r_{Tsi} + Z_{si} r_{Csi}$ denote the
counterfactual outcome. Call an outcome of $1$ a \emph{success} and
$0$ a \emph{failure}.

Define the treatment effect $\tau_{si} := r_{Tsi} - r_{Csi} \in \{-1, 0, 1\}$.
The average treatment effect over all $2S$ units is
\begin{displaymath}
   \ATE := \frac{1}{2S} \sum_{s=1}^S \sum_{i=1}^2 \tau_{si}.
\end{displaymath}
The vector $\tau$ is a fixed feature of the finite population. The
randomness in the problem comes solely from the treatment assignment
or selection mechanism that determines $Z$.

It is convenient to split the $\ATE$ into two random pieces, the
\emph{attributable effects} of
\citet{rosenbaum-2001-effects-attributable-to-treatment, rosenbaum-2002-attributing-effects-to-treatment-in-matched-studies}
and \citet{rigdon-2015-attributable-effect}:
\begin{displaymath}
   \Ao(Z, \tau) := \sum_{s, i} Z_{si} \cdot \tau_{si},
   \qquad
   \Az(Z, \tau) := \sum_{s, i} (1 - Z_{si}) \cdot \tau_{si}.
\end{displaymath}
$\Ao$ is the \emph{net treatment effect} among treated units: the
number of successes observed among treated units the treatment
\emph{caused} minus the number of successes the treatment
\emph{prevented}. $\Az$ is the corresponding quantity for the
untreated. Both are random quantities because they depend on the
treatment assignments, $Z$; however,
\begin{equation}\label{eqn:sum-constraint}
   \Ao(Z, \tau) + \Az(Z, \tau) = \sum_{s, i} \tau_{si} = 2S \cdot \ATE,
\end{equation}
a non-random quantity. Following Rigdon and Hudgens, we call the sets
we construct for $\Ao$ and $\Az$ \emph{prediction sets}, not
confidence sets, because the targets are random. As a parameter, the
$\ATE$ has a confidence set we will build from the two prediction sets
in \S\ref{sec:ate}.

Routine algebra rewrites $\Ao$ in terms of observed and counterfactual
outcomes:
\begin{displaymath}
   \Ao = \sum_{s, i} Z_{si} \cdot R_{si} \cdot (1 - R_{si}')
       - \sum_{s, i} Z_{si} \cdot R_{si}' \cdot (1 - R_{si}).
\end{displaymath}
The non-zero contributions to the first term are the treated units
with observed successes ($R_{si} = 1$) that would have been failures
had they not been treated ($R_{si}' = 0$); these are successes
\emph{caused} by treatment. The second term counts treated successes
\emph{prevented} by treatment (observed failures whose counterfactual
was a success). Under the \emph{monotonicity assumption}
($\tau_{si} \geq 0$ for all units), the second term is zero. The
monotonicity assumption is plausible in many circumstances, but is
unnecessary for inferences regarding net treatment effects.

Inference regarding the effect of treatment on the treated units
involves speculation about the counterfactual (control) outcomes for
these units. We do not need to speculate about the treated outcomes
for these units, because we observe them. Thus, inference about $\Ao$
focuses on $r_C$. Similarly, inference about $\Az$ focuses on $r_T$.
We provide detailed derivations only for $\Ao$ to keep the discussion
efficient; parallel arguments apply to $\Az$.

A sharp null hypothesis specifies a pattern of treatment effects,
$ \{ \tau_{si} \}, $ and implies the counterfactual response for each
unit. The data rule out many sharp nulls $\tau^0$ before any
formal test. If a treated unit has $R_{si} = 0$, then $r_{Tsi} = 0$
and so $\tau_{si} \neq 1$. The general statement is
\begin{equation}\label{eqn:compatible}
   \tau_{si} \in \{(2Z_{si} - 1) \cdot R_{si},\;
                   (2Z_{si} - 1) \cdot (R_{si} - 1)\}
\end{equation}
for each unit. The set in~(\ref{eqn:compatible}) is \emph{observed},
so we can immediately reject certain patterns of treatment effects as
incompatible with the data. If~(\ref{eqn:compatible}) holds for the
pattern asserted by a particular sharp null hypothesis, call it a
\emph{compatible} hypothesis. There are at most $2^{2S}$ compatible
sharp null hypotheses.

When testing hypotheses regarding $\Ao$, two compatible sharp null
hypotheses specifying the same pattern of effects for the treated
units but different patterns for the control units yield the same
pattern of control potential outcomes. Any test based on $r_C$ yields
the same inference for two such hypotheses. This reduces the number of
distinct tests to $2^S.$

\paragraph{McNemar's test and equivalence classes.}
We use McNemar's test for matched pairs with binary outcomes. Consider
the pairs of control potential outcomes $(r_{Cs1}, r_{Cs2})$ (which
includes the counterfactual outcomes for treated units, as implied by
the null hypothesis). Some of these pairs consist of two successes or
two failures; these are \emph{concordant} pairs. Other pairs consist
of one success and one failure, the \emph{discordant} pairs. In an
experiment (or a strongly ignorable observational study), the unit
selected for treatment from a discordant pair is equally likely to
have a success or failure as the control outcome. The number of
treated units among discordant pairs with a success as the control
outcome thus has a $\Binom(S^\ast, \tfrac{1}{2})$ distribution, where
$S^\ast$ is the number of discordant pairs. McNemar's test rejects
when the number of such units is too large or too small compared to
this reference distribution.

Categorize pairs by their pattern of observed outcomes. Let $\Sjk$
count pairs whose treated unit has $R = j$ and whose control unit has
$R = k$, for $j, k \in \{0, 1\}$. Table~\ref{tbl:treatment-concordant}
records, for each pattern, the value of $\tau_{si}^0$ for the treated
unit that would make the pair discordant according to the test
statistic. For example, in a pair having a treated success and a
control failure, if the null hypothesis asserts no treatment effect
for the treated unit, its counterfactual outcome would still have been
a success, and we would consider the pair discordant.

\begin{table}[ht]
\centering
\caption{Hypothesized treatment effect on the treated unit
that renders a pair concordant or discordant under the implied
control outcomes.}
\label{tbl:treatment-concordant}
\begin{tabular}{cccccl}
\toprule
$R_{si_T}$ & $R_{si_C}$ & Concordant & Discordant & Pairs \\
\midrule
0 & 0 & $0$  & $-1$ & $\Szz$ \\
0 & 1 & $-1$ & $0$  & $\Szo$ \\
1 & 0 & $1$  & $0$  & $\Soz$ \\
1 & 1 & $0$  & $1$  & $\Soo$ \\
\bottomrule
\end{tabular}
\end{table}

For a hypothesized $\tau^0$, let $\mjk \in \{0, \ldots, \Sjk\}$
count the pairs of category $(j,k)$ the test statistic deems
discordant. The total number of discordant pairs is
$S^\ast = \sum_{j,k} \mjk$, and the number of those discordant pairs
in which the treated unit is the success-under-control is
$\mpz := \mzz + \moz$. Symmetrically, the number whose control
unit is the success-under-control is $\mpo := \mzo + \moo$.

McNemar's statistic depends on $\tau^0$ only through $(\mpz, \mpo)$:
under the null hypothesis,
\begin{displaymath}
  \mpz \sim \Binom(\mpz + \mpo, \tfrac{1}{2}).
\end{displaymath}
Two null hypotheses with the same $(\mpz, \mpo)$ lead to the same
inference. The number of effectively distinct inferences is therefore
\begin{equation}\label{eqn:eq-count}
   (\Szz + \Soz + 1)\cdot(\Szo + \Soo + 1),
\end{equation}
quadratic in $S$. The sharp null hypotheses thus fall into
\emph{equivalence classes} indexed by $(\mpz, \mpo)$: two sharp nulls
in the same equivalence class are either both rejected, or both
retained.

\paragraph{The role of monotonicity}
When the monotonicity assumption holds, it implies $ \mzz = 0 $
and $ \mzo = \Szo, $ which in turn implies there are only
$(\Soz + 1) \cdot (\Soo + 1)$ distinct inferences.

\section{Composite nulls and consistency constraints}\label{sec:nulls}
Often we are less interested in sharp null hypotheses than composite
nulls regarding net treatment effects:
\begin{itemize}
  \item $\Hz: \Ao \leq \az$ vs $\Ho: \Ao > \az$
  \item $\Hz^\prime: \Ao \geq \az$ vs $\Ho^\prime: \Ao < \az$
  \item $\Hz^{\dagger}: \Ao = \az$ vs $\Ho^{\dagger}: \Ao \neq \az$
\end{itemize}
The first two hypotheses are one-sided; the two-sided third hypothesis
may be tested by performing the two one-sided tests at half the
nominal level to adjust for multiple comparisons
\citep[\S2.2]{rosenbaum-2015-two-r-packages}. Subsequent discussion
will focus on $\Hz$ but the extension to other hypotheses is
straightforward.

\paragraph{Bounds on the inferential target.}
Let $f_T$ and $s_T$ denote the numbers of failures and successes
among the treated units. We know a priori $-f_T \leq \Ao \leq s_T$.
The two extremes correspond to the scenario where every treated failure
was caused by treatment and no successes were, and to the scenario
where every treated success was caused by treatment and no failures
were. This allows us to reject with logical certainty hypotheses
$\Hz: \Ao \leq \az$ for $ \az < -f_T $ or $\Hz^\prime: \Ao \geq \az$
for $ \az > s_T. $ Other hypotheses we know to be true a priori:
$\Hz: \Ao \leq \az$ for $ \az \geq s_T $ or $\Hz^\prime: \Ao \geq \az$
for $ \az \leq -f_T. $ We restrict the range of
nontrivial $\az$ to those in $[-f_T, s_T)$ for $ \Hz $
and $\az \in (-f_T, s_T]$ for $ \Hz^\prime $.

\paragraph{Consistency constraints.}
The realized $\Ao$ for a hypothesized $\tau^0$ has the closed form
\begin{align}
  \Ao &= \Soz - \moz + \moo - \mzz - \Szo + \mzo \nonumber \\
       &= \Soz - \Szo + \mpo - \mpz. \nonumber
\end{align}
A pattern $(\mpzo, \mpoo)$ is \emph{consistent} with $\Hz: \Ao \leq \az$
if and only if
\begin{displaymath}
   \mpzo - \mpoo \geq \Soz - \Szo - \az =: \deltaz,
\end{displaymath}
subject to the box constraints
$0 \leq \mpzo \leq \Szz + \Soz$ and
$0 \leq \mpoo \leq \Szo + \Soo$. When monotonicity holds, the box
constraints become $0 \leq \mpzo \leq \Soz$ and
$\Szo \leq \mpoo \leq \Szo + \Soo$. We reject a composite null
only if we reject every sharp null compatible with it; this is the
dictum of \citet[pg. 290]{rosenbaum-2019-observation-experiment} that
``to reject a composite hypothesis is to reject each and every way it
may be true.'' Rejecting $\Hz$ requires rejecting every consistent
$(\mpzo, \mpoo)$.

For moderate $S$, it would be feasible to check each pattern $(\mpz,
\mpo)$ corresponding to the composite null. We reject the composite
if and only if we reject each such pattern. In the next section, we
identify a single pattern hardest to reject. If we reject this
pattern, we know all other consistent patterns, and thus the composite
hypothesis, are also rejected.

\paragraph{Example} Suppose we have a controlled experiment involving
$S=1000$ pairs, with one unit in each pair selected at random for
treatment. In $\Szz = 800$ of the pairs, neither unit experienced a
successful outcome; in $\Soo = 100$ of the pairs, both units did. In
$\Soz = 70$ of the pairs, the treated unit experienced a success
while the control unit did not. In the remaining $\Szo = 30$ pairs,
the control unit experienced a success while the treated unit did not.
We observe the $ \Sjk $ counts.

Suppose we wish to test $\Hz: \Ao \leq 0, $ with $\az = 0$. Rejecting
this null hypothesis would constitute evidence the net treatment
effect was positive. Then $ \deltaz = 70 - 30 - 0 = 40. $ A hypothesis
$ \tau^0 $ corresponds to a claim about how many of the $ \Sjk $ pairs
would be discordant if we observed the control outcomes for treated
units. A particular $ \tau^0 $ might correspond to $ \mzz = 400, $
$ \moo = 50, $ $ \moz = 35, $ and $ \mzo = 15; $ that is, half the
pairs in each category would be discordant according to the control
outcomes. Then $ \mpz = 435 $ and $ \mpo = 65. $ This pattern is
consistent with $ \Hz, $ since the realized $\Ao = -330 \leq 0$
(equivalently, since $ 435 - 65 \geq 40 $) and so it is one of the
patterns we must test if we hope to reject $ \Hz. $

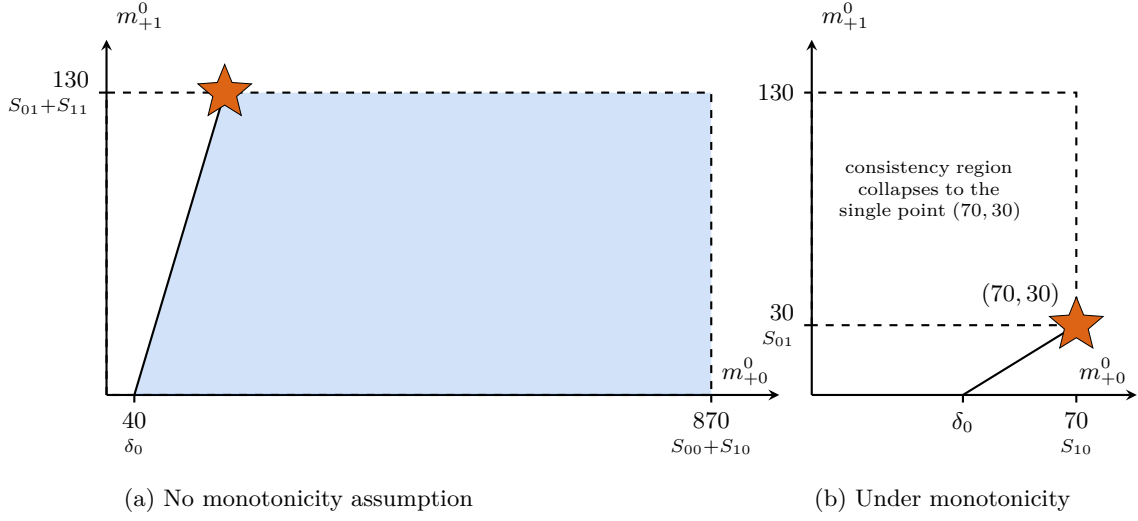
\begin{figure}[t]
\centering
\begin{minipage}[t]{0.48\linewidth}
\centering
\begin{tikzpicture}[>=stealth]
  \draw[->,thick] (0,0) -- (8.9,0) node[anchor=south east,font=\small] {$m_{+0}^0$};
  \draw[->,thick] (0,0) -- (0,4.7) node[anchor=south west,font=\small] {$m_{+1}^0$};
  \fill[metablue,opacity=0.18] (0.3678,0) -- (8.0,0) -- (8.0,4.0) -- (1.5632,4.0) -- cycle;
  \draw[dashed,thick] (0,0) rectangle (8.0,4.0);
  \draw[thick] (0.3678,0) -- (1.5632,4.0);
  \draw[thick] (8.0,0) -- (8.0,-0.1) node[anchor=north,align=center,font=\small]
      {$870$\\[-2pt]{\scriptsize$S_{00}{+}S_{10}$}};
  \draw[thick] (0,4.0) -- (-0.1,4.0) node[anchor=east,align=right,font=\small]
      {$130$\\[-2pt]{\scriptsize$S_{01}{+}S_{11}$}};
  \draw[thick] (0.3678,0) -- (0.3678,-0.1) node[anchor=north,align=center,font=\small]
      {$40$\\[-2pt]{\scriptsize$\delta_0$}};
  \node[star,star points=5,star point ratio=2.25,minimum size=10pt,
        fill=metaorange,draw=black,line width=0.3pt] at (1.5632,4.0) {};
\end{tikzpicture}\\[4pt]
{\small (a) No monotonicity assumption}
\end{minipage}\hfill
\begin{minipage}[t]{0.48\linewidth}
\centering
\begin{tikzpicture}[>=stealth]
  \draw[->,thick] (0,0) -- (4.3,0) node[anchor=south east,font=\small] {$m_{+0}^0$};
  \draw[->,thick] (0,0) -- (0,4.7) node[anchor=south west,font=\small] {$m_{+1}^0$};
  \draw[dashed,thick] (0,0.9231) rectangle (3.5,4.0);
  \draw[thick] (2.0,0) -- (3.5,0.9231);
  \draw[thick] (3.5,0) -- (3.5,-0.1) node[anchor=north,align=center,font=\small]
      {$70$\\[-2pt]{\scriptsize$S_{10}$}};
  \draw[thick] (2.0,0) -- (2.0,-0.1) node[anchor=north,font=\small] {$\delta_0$};
  \draw[thick] (0,0.9231) -- (-0.1,0.9231) node[anchor=east,align=right,font=\small]
      {$30$\\[-2pt]{\scriptsize$S_{01}$}};
  \draw[thick] (0,4.0) -- (-0.1,4.0) node[anchor=east,font=\small] {$130$};
  \node[star,star points=5,star point ratio=2.25,minimum size=10pt,
        fill=metaorange,draw=black,line width=0.3pt] at (3.5,0.9231) {};
  \node[anchor=south east,font=\small,inner sep=3pt] at (3.4,1.05) {$(70,30)$};
  \node[align=center,font=\scriptsize,text width=3.2cm] at (1.55,2.7)
      {consistency region\\ collapses to the\\ single point $(70,30)$};
\end{tikzpicture}\\[4pt]
{\small (b) Under monotonicity}
\end{minipage}
\caption{The set of patterns $(\mpzo, \mpoo)$ consistent with the
  composite null $\Hz: \Ao \leq 0$ in the example, where
  $\deltaz = \Soz - \Szo - \az = 40$. (a) Without the monotonicity
  assumption the consistent patterns form the shaded trapezoid,
  bounded to the left by the consistency line
  $\mpoo = \mpzo - \deltaz$ and on the other sides by the box
  constraints. We reject the composite null only if we reject every
  pattern in the region. (b) Under monotonicity the box shrinks to
  $\mpzo \in [0, \Soz]$ and $\mpoo \in [\Szo, \Szo + \Soo]$, and the
  consistent region collapses to the single point $(70, 30)$.
  \S\ref{sec:hardest} discusses the starred pattern.}
\label{fig:consistency}
\end{figure}

Although 100 of the pairs are discordant according to the
\emph{observed} responses, McNemar uses the implied control outcomes
as the basis for the test. In this example, the null hypothesis
corresponds to $S^\ast = 500$ discordant pairs, and the null
hypothesis would have us believe 435 of these pairs have the
treated unit as the success-under-control. Since
$ \Pr\{\Binom(500, \tfrac{1}{2}) \geq 435 \} \approx 1.52 \times 10^{-68}$,
McNemar's test would reject this specific pattern. Yet there are many
other patterns $(\mpzo, \mpoo)$ consistent with $ \Hz: \Ao \leq 0, $
so rejecting this particular pattern is just the first step. The full
consistent set is the shaded trapezoid of
Figure~\ref{fig:consistency}a.

\paragraph{The role of monotonicity}
Under the monotonicity assumption, $ \mzz = 0 $ and
$ \mzo = \Szo $ (\S\ref{sec:setting}), so the consistent patterns
are confined to $ \mpzo \in [0, \Soz] $ and
$ \mpoo \in [\Szo, \Szo + \Soo]. $ Therefore,
$ \mpzo - \mpoo \leq \Soz - \Szo.$ For the null $\Hz: \Ao \leq 0$
(where $ \deltaz = \Soz - \Szo $) the consistency constraint is
$\mpzo - \mpoo \geq \Soz - \Szo, $ so the only possibility is
$ (\mpzo, \mpoo) = (\Soz, \Szo) $ (Figure~\ref{fig:consistency}b).
This is the textbook application of McNemar's test, which typically
considers only the sharp null hypothesis of no treatment effect for
any unit \citep{lehmann-1975-nonparametrics}. In general, the
consistent region may include more than one pattern we must test, but
the next section shows we can identify the single hardest pattern to
reject.

\section{The single hardest pattern}\label{sec:hardest}
For a consistent $\tau^0$, the McNemar test statistic is $\mpzo$,
and its null distribution under strongly ignorable treatment selection is
$\Binom(\mpzo + \mpoo, \tfrac{1}{2})$. Consistency forces $\mpzo$
above the binomial mean (when $\deltaz > 0$), so the natural test is
right-tailed:
\begin{equation}\label{eqn:pgreater}
   \pgt(a, b) := \Pr\{\Binom(a + b, \tfrac{1}{2}) \geq a\},
   \qquad a := \mpzo,\ b := \mpoo.
\end{equation}
We reject $\Hz$ if and only if
$\sup\{\pgt(a, b) : (a, b) \, \textrm{consistent} \} \leq \alpha$.

The supremum admits a closed form. Two monotonicity properties of
$\pgt$ identify the pattern hardest to reject.

\begin{lemma}[Pointwise monotonicity in $a$]\label{lem:mono-a}
For each fixed $b$ and $a \geq 0$,
$\pgt(a + 1, b) < \pgt(a, b)$.
\end{lemma}
The practical consequence of Lemma~\ref{lem:mono-a} is that the
pattern hardest to reject lies on the left boundary of the consistent
region as illustrated in Figure~\ref{fig:consistency}.

\begin{proof}
Suppose $ B_1 \sim \Binom(a+b, \tfrac{1}{2}) $ and $ B_2 \sim \mathrm{Bernoulli}(\tfrac{1}{2}). $
Then $ B_3 = B_1 + B_2 $ has a $ \Binom(a+b+1, \tfrac{1}{2}) $
distribution. By iterated expectations,
\begin{align}
  \pgt(a + 1, b) &= \Pr\{ B_3 \geq a + 1\} \nonumber \\
  &= \Pr\{ B_1 + B_2 \geq a + 1 \mid B_2 = 0\} \cdot \Pr\{B_2 = 0\} \nonumber \\
  &\phantom{=} \hspace{10pt} + \Pr\{ B_1 + B_2 \geq a + 1 \mid B_2 = 1\} \cdot \Pr\{B_2 = 1\} \nonumber \\
  &= \tfrac{1}{2} \cdot \Pr\{ B_1 \geq a + 1 \} + \tfrac{1}{2} \cdot \Pr\{ B_1 \geq a \} \nonumber \\
  &= \tfrac{1}{2} \cdot \left[ \Pr\{ B_1 \geq a \} - \Pr\{ B_1 = a \} \right] + \tfrac{1}{2} \cdot \Pr\{ B_1 \geq a \} \nonumber \\
  &= \Pr\{ B_1 \geq a \} - \tfrac{1}{2} \Pr\{ B_1 = a \} \nonumber \\
  &= \pgt(a, b) - \tfrac{1}{2}\Pr\{ B_1 = a\} \nonumber \\
  &< \pgt(a, b), \nonumber
\end{align}
since the binomial point mass at $a$ is strictly positive. \qedhere
\end{proof}

\paragraph{Example, continued.} Fix the $\mpo = 65$ slice. The
example pattern sits at $\mpz = 435$, with
$\pgt(435, 65) \approx 1.52\times 10^{-68}$. Stepping one unit to
the right gives $\pgt(436, 65) \approx 8.73\times 10^{-69}$, just
under half the previous value, in agreement with
Lemma~\ref{lem:mono-a}. The full slice runs over
$\mpz \in [105, 870]$: the lower endpoint is the consistency floor
$\mpo + \deltaz = 65 + 40$; the upper endpoint is the box-constraint
ceiling $\Szz + \Soz = 870$, above which the pattern is
incompatible with the observed counts. Figure~\ref{fig:lemma1} plots
$\pgt(\mpz, 65)$ over this range on a log scale.

\begin{figure}[ht]
\centering
\begin{tikzpicture}
  \draw[thick] (0,0) rectangle (10.5,5.0);
  \draw[gray!25, very thin] (0,5.0000) -- (10.5,5.0000);
  \draw[thick] (0,5.0000) -- (-0.10,5.0000) node[anchor=east,font=\small] {$10^{0}$};
  \draw[gray!25, very thin] (0,3.7500) -- (10.5,3.7500);
  \draw[thick] (0,3.7500) -- (-0.10,3.7500) node[anchor=east,font=\small] {$10^{-50}$};
  \draw[gray!25, very thin] (0,2.5000) -- (10.5,2.5000);
  \draw[thick] (0,2.5000) -- (-0.10,2.5000) node[anchor=east,font=\small] {$10^{-100}$};
  \draw[gray!25, very thin] (0,1.2500) -- (10.5,1.2500);
  \draw[thick] (0,1.2500) -- (-0.10,1.2500) node[anchor=east,font=\small] {$10^{-150}$};
  \draw[gray!25, very thin] (0,0.0000) -- (10.5,0.0000);
  \draw[thick] (0,0.0000) -- (-0.10,0.0000) node[anchor=east,font=\small] {$10^{-200}$};
  \draw[thick] (0.0677,0) -- (0.0677,-0.10) node[anchor=north,font=\small] {$105$};
  \draw[thick] (1.3548,0) -- (1.3548,-0.10) node[anchor=north,font=\small] {$200$};
  \draw[thick] (4.0645,0) -- (4.0645,-0.10) node[anchor=north,font=\small] {$400$};
  \draw[thick] (6.7742,0) -- (6.7742,-0.10) node[anchor=north,font=\small] {$600$};
  \draw[thick] (10.4323,0) -- (10.4323,-0.10) node[anchor=north,font=\small] {$870$};
  \node[anchor=north,font=\small] at (5.25, -0.65) {$\mpz$ (with $\mpo = 65$)};
  \node[anchor=south,rotate=90,font=\small] at (-1.15, 2.5) {$p_{>}(\mpz, 65)$};
  \draw[metaorange, dashed, thick] (0,4.9675) -- (10.5,4.9675);
  \node[metaorange, anchor=south east, font=\small, inner sep=2pt] at (10.45,4.9675) {$\alpha=0.05$};
  \draw[metablue, thick] (0.0677,4.9281) -- (0.1084,4.9206) -- (0.1490,4.9128) -- (0.1897,4.9047) -- (0.2303,4.8963) -- (0.2710,4.8876) -- (0.3116,4.8787) -- (0.3523,4.8695) -- (0.3929,4.8601) -- (0.4335,4.8504) -- (0.4742,4.8405) -- (0.5148,4.8304) -- (0.5555,4.8200) -- (0.5961,4.8095) -- (0.6368,4.7987) -- (0.6774,4.7878) -- (0.7181,4.7766) -- (0.7587,4.7653) -- (0.7994,4.7538) -- (0.8400,4.7422) -- (0.8806,4.7304) -- (0.9213,4.7184) -- (0.9619,4.7062) -- (1.0026,4.6939) -- (1.0432,4.6815) -- (1.0839,4.6689) -- (1.1245,4.6562) -- (1.1652,4.6433) -- (1.2058,4.6303) -- (1.2465,4.6172) -- (1.2871,4.6040) -- (1.3277,4.5906) -- (1.3684,4.5771) -- (1.4090,4.5635) -- (1.4497,4.5498) -- (1.4903,4.5360) -- (1.5310,4.5220) -- (1.5716,4.5080) -- (1.6123,4.4939) -- (1.6529,4.4796) -- (1.6935,4.4653) -- (1.7342,4.4508) -- (1.7748,4.4363) -- (1.8155,4.4217) -- (1.8561,4.4070) -- (1.8968,4.3922) -- (1.9374,4.3773) -- (1.9781,4.3624) -- (2.0187,4.3473) -- (2.0594,4.3322) -- (2.1000,4.3170) -- (2.1406,4.3017) -- (2.1813,4.2864) -- (2.2219,4.2710) -- (2.2626,4.2555) -- (2.3032,4.2399) -- (2.3439,4.2243) -- (2.3845,4.2086) -- (2.4252,4.1928) -- (2.4658,4.1770) -- (2.5065,4.1611) -- (2.5471,4.1451) -- (2.5877,4.1291) -- (2.6284,4.1130) -- (2.6690,4.0969) -- (2.7097,4.0807) -- (2.7503,4.0644) -- (2.7910,4.0481) -- (2.8316,4.0317) -- (2.8723,4.0153) -- (2.9129,3.9988) -- (2.9535,3.9823) -- (2.9942,3.9657) -- (3.0348,3.9491) -- (3.0755,3.9324) -- (3.1161,3.9157) -- (3.1568,3.8989) -- (3.1974,3.8821) -- (3.2381,3.8652) -- (3.2787,3.8483) -- (3.3194,3.8313) -- (3.3600,3.8143) -- (3.4006,3.7973) -- (3.4413,3.7802) -- (3.4819,3.7631) -- (3.5226,3.7459) -- (3.5632,3.7287) -- (3.6039,3.7114) -- (3.6445,3.6941) -- (3.6852,3.6768) -- (3.7258,3.6594) -- (3.7665,3.6420) -- (3.8071,3.6245) -- (3.8477,3.6070) -- (3.8884,3.5895) -- (3.9290,3.5719) -- (3.9697,3.5544) -- (4.0103,3.5367) -- (4.0510,3.5191) -- (4.0916,3.5014) -- (4.1323,3.4836) -- (4.1729,3.4659) -- (4.2135,3.4481) -- (4.2542,3.4302) -- (4.2948,3.4124) -- (4.3355,3.3945) -- (4.3761,3.3765) -- (4.4168,3.3586) -- (4.4574,3.3406) -- (4.4981,3.3226) -- (4.5387,3.3045) -- (4.5794,3.2865) -- (4.6200,3.2684) -- (4.6606,3.2502) -- (4.7013,3.2321) -- (4.7419,3.2139) -- (4.7826,3.1957) -- (4.8232,3.1774) -- (4.8639,3.1592) -- (4.9045,3.1409) -- (4.9452,3.1226) -- (4.9858,3.1042) -- (5.0265,3.0859) -- (5.0671,3.0675) -- (5.1077,3.0490) -- (5.1484,3.0306) -- (5.1890,3.0121) -- (5.2297,2.9936) -- (5.2703,2.9751) -- (5.3110,2.9566) -- (5.3516,2.9380) -- (5.3923,2.9194) -- (5.4329,2.9008) -- (5.4735,2.8822) -- (5.5142,2.8636) -- (5.5548,2.8449) -- (5.5955,2.8262) -- (5.6361,2.8075) -- (5.6768,2.7887) -- (5.7174,2.7700) -- (5.7581,2.7512) -- (5.7987,2.7324) -- (5.8394,2.7136) -- (5.8800,2.6948) -- (5.9206,2.6759) -- (5.9613,2.6570) -- (6.0019,2.6381) -- (6.0426,2.6192) -- (6.0832,2.6003) -- (6.1239,2.5813) -- (6.1645,2.5624) -- (6.2052,2.5434) -- (6.2458,2.5244) -- (6.2865,2.5053) -- (6.3271,2.4863) -- (6.3677,2.4672) -- (6.4084,2.4481) -- (6.4490,2.4290) -- (6.4897,2.4099) -- (6.5303,2.3908) -- (6.5710,2.3717) -- (6.6116,2.3525) -- (6.6523,2.3333) -- (6.6929,2.3141) -- (6.7335,2.2949) -- (6.7742,2.2757) -- (6.8148,2.2564) -- (6.8555,2.2372) -- (6.8961,2.2179) -- (6.9368,2.1986) -- (6.9774,2.1793) -- (7.0181,2.1600) -- (7.0587,2.1406) -- (7.0994,2.1213) -- (7.1400,2.1019) -- (7.1806,2.0826) -- (7.2213,2.0632) -- (7.2619,2.0437) -- (7.3026,2.0243) -- (7.3432,2.0049) -- (7.3839,1.9854) -- (7.4245,1.9660) -- (7.4652,1.9465) -- (7.5058,1.9270) -- (7.5465,1.9075) -- (7.5871,1.8880) -- (7.6277,1.8684) -- (7.6684,1.8489) -- (7.7090,1.8293) -- (7.7497,1.8098) -- (7.7903,1.7902) -- (7.8310,1.7706) -- (7.8716,1.7510) -- (7.9123,1.7314) -- (7.9529,1.7117) -- (7.9935,1.6921) -- (8.0342,1.6724) -- (8.0748,1.6528) -- (8.1155,1.6331) -- (8.1561,1.6134) -- (8.1968,1.5937) -- (8.2374,1.5740) -- (8.2781,1.5542) -- (8.3187,1.5345) -- (8.3594,1.5147) -- (8.4000,1.4950) -- (8.4406,1.4752) -- (8.4813,1.4554) -- (8.5219,1.4356) -- (8.5626,1.4158) -- (8.6032,1.3960) -- (8.6439,1.3762) -- (8.6845,1.3563) -- (8.7252,1.3365) -- (8.7658,1.3166) -- (8.8065,1.2967) -- (8.8471,1.2769) -- (8.8877,1.2570) -- (8.9284,1.2371) -- (8.9690,1.2172) -- (9.0097,1.1972) -- (9.0503,1.1773) -- (9.0910,1.1574) -- (9.1316,1.1374) -- (9.1723,1.1175) -- (9.2129,1.0975) -- (9.2535,1.0775) -- (9.2942,1.0575) -- (9.3348,1.0375) -- (9.3755,1.0175) -- (9.4161,0.9975) -- (9.4568,0.9775) -- (9.4974,0.9574) -- (9.5381,0.9374) -- (9.5787,0.9173) -- (9.6194,0.8973) -- (9.6600,0.8772) -- (9.7006,0.8571) -- (9.7413,0.8370) -- (9.7819,0.8169) -- (9.8226,0.7968) -- (9.8632,0.7767) -- (9.9039,0.7566) -- (9.9445,0.7364) -- (9.9852,0.7163) -- (10.0258,0.6961) -- (10.0665,0.6760) -- (10.1071,0.6558) -- (10.1477,0.6356) -- (10.1884,0.6155) -- (10.2290,0.5953) -- (10.2697,0.5751) -- (10.3103,0.5549) -- (10.3510,0.5347) -- (10.3916,0.5144) -- (10.4323,0.4942);
  \filldraw[black] (0.0677,4.9281) circle (1.6pt);
  \draw[->, thick, black] (0.5,3.6) -- (0.085,4.87);
  \node[anchor=west, font=\small, inner sep=3pt] at (0.5,3.5) {$(105,\,1.3\!\times\!10^{-3})$};
  \filldraw[black] (4.5387,3.3045) circle (1.6pt);
  \draw[->, thick, black] (6.7742,4.2500) -- (4.5887,3.3545);
  \node[anchor=west, font=\small, inner sep=2pt] at (6.7742,4.2500) {$\mpz = 435$ (example pattern)};
\end{tikzpicture}
\caption{Lemma~\ref{lem:mono-a} in the running example: along the
$\mpo = 65$ slice, $\pgt(\mpz, 65)$ is strictly decreasing in
$\mpz$. The supremum over consistent values is therefore attained
at the leftmost endpoint $\mpz = \mpo + \deltaz = 105$.}
\label{fig:lemma1}
\end{figure}
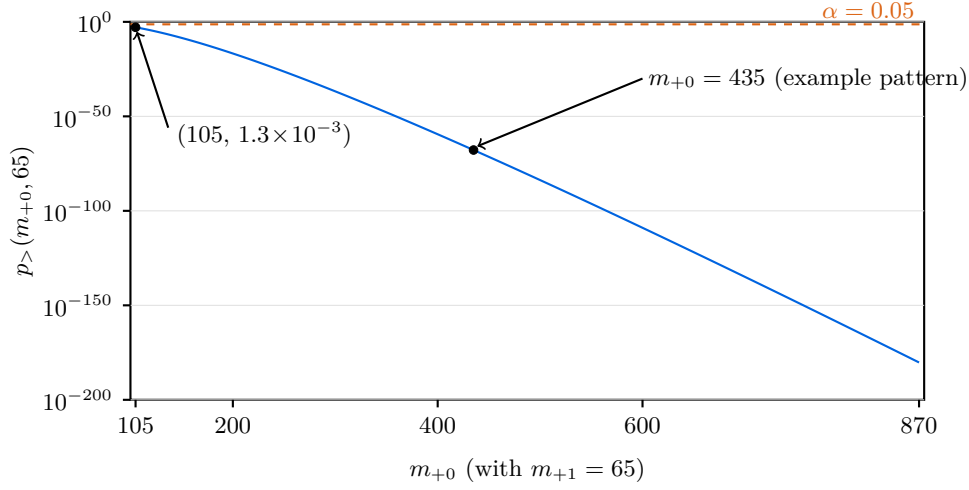

Because $\pgt(\cdot, 65)$ is monotone decreasing, the supremum on this
slice occurs at the smallest consistent value $\mpz = 105$, where
$\pgt(105, 65) \approx 1.34\times 10^{-3}$. This single value is below
conventional $\alpha$, so we reject every pattern with $\mpo = 65$; we
need not evaluate any of the other 765 points on this slice. By
Lemma~\ref{lem:mono-a} the same holds for every $\mpo$, and the
supremum of $\pgt$ over consistent $(a, b)$ lies on the boundary
$a = b + \deltaz$. We are left with a one-dimensional problem in $b$.

Define $ \pgt^\prime(b) := \pgt(b + \deltaz, b) = \Pr\{ \Binom(2b + \deltaz, \tfrac{1}{2}) \geq b + \deltaz \}.$
Lemma~\ref{lem:mono-b} demonstrates that $ \pgt^\prime $ is increasing
in $ b $ when $ \deltaz \geq 2 $ and decreasing when $ \deltaz = 0. $
When $ \deltaz = 1, $ $ \pgt^\prime(b) = \tfrac{1}{2} $ for all
$ 0 \leq b \leq \Szo + \Soo. $ Since $ \deltaz = \Soz - \Szo - \az $
is known as soon as we specify $ \az, $ Lemma~\ref{lem:mono-b} allows
us to identify which boundary value will yield the largest p-value.
The case $ \deltaz < 0 $ (which arises when $ \az > \Soz - \Szo $)
falls outside the lemma's scope and is discussed in \S\ref{sec:practical}.

\begin{lemma}[Boundary monotonicity in $b$]\label{lem:mono-b}
For integer $\deltaz \geq 0$ and integer $b \geq 0$,
\begin{displaymath}
  \pgt^\prime(b + 1) = \pgt^\prime(b) + \frac{1}{4} \cdot \frac{\deltaz - 1}{b + \deltaz} \cdot \Pr\{\Binom(2b + \deltaz,\, \tfrac{1}{2}) = b + \deltaz - 1\}.
\end{displaymath}
The boundary value is strictly increasing in $b$ for
$\deltaz \geq 2$, identically $\tfrac{1}{2}$ for $\deltaz = 1$, and
strictly decreasing for $\deltaz = 0$.
\end{lemma}

The practical consequence of Lemma~\ref{lem:mono-b} is the
pattern hardest to reject lies at the upper boundary of the consistent
region (when $\deltaz \geq 2$) as illustrated in
Figure~\ref{fig:consistency}, and when combined with
Lemma~\ref{lem:mono-a}, uniquely identifies the single hardest to
reject pattern. If we reject that pattern, we reject $ \Hz $.

\begin{proof}
Suppose $ B_1 \sim \Binom(2b + \deltaz, \tfrac{1}{2}) $ and
$ B_2 \sim \Binom(2, \tfrac{1}{2}). $ Then $B_3 = B_1 + B_2$ has a
$ \Binom(2(b + 1) + \deltaz, \tfrac{1}{2}) $ distribution. By
iterated expectations,
\begin{align}
  \pgt^\prime(b+1) &= \Pr\{ B_3 \geq b + 1 + \deltaz \} \nonumber \\
  &= \Pr\{ B_1 \geq b + 1 + \deltaz \} \cdot \Pr\{ B_2 = 0 \} \nonumber \\
  &\phantom{=} \hspace{10pt} + \Pr\{ B_1 \geq b + \deltaz \} \cdot \Pr\{ B_2 = 1 \} \nonumber \\
  &\phantom{=} \hspace{10pt} + \Pr\{ B_1 \geq b - 1 + \deltaz \} \cdot \Pr\{ B_2 = 2 \} \nonumber \\
  &= \frac{1}{4} \Pr\{ B_1 \geq b + 1 + \deltaz \}
     + \frac{1}{2} \Pr\{ B_1 \geq b + \deltaz \}
     + \frac{1}{4} \Pr\{ B_1 \geq b - 1 + \deltaz \} \nonumber \\
  &= \Pr\{ B_1 \geq b + \deltaz \} + \frac{1}{4}\left[ \Pr\{ B_1 = b - 1 + \deltaz \} - \Pr\{ B_1 = b + \deltaz \} \right] \nonumber \\
  &= \pgt^\prime(b) + \frac{1}{4} \cdot \left(1 - \frac{b+1}{b+\deltaz}\right) \Pr\{ B_1 = b - 1 + \deltaz \} \nonumber \\
  &= \pgt^\prime(b) + \frac{1}{4} \cdot \frac{\deltaz - 1}{b+\deltaz} \cdot \Pr\{ B_1 = b - 1 + \deltaz \}. \nonumber
\end{align}
The second-to-last line uses the Binomial probability mass function to
write $ \Pr\{B_1 = b + \deltaz \} = \frac{b+1}{b+\deltaz} \cdot \Pr\{B_1 = b - 1 + \deltaz \}.$

When $ \deltaz \geq 2, $ we have $ \pgt^\prime(b+1) > \pgt^\prime(b), $
so $ \pgt^\prime $ is increasing. When $ \deltaz = 1, $ $ \pgt^\prime(b+1) = \pgt^\prime(b), $
so $ \pgt^\prime $ is a constant. Since
$ \pgt^\prime(b) = \Pr\{ \Binom(2b + 1, \tfrac{1}{2}) \geq b + 1 \}, $
and since $ \Binom(2b + 1, \tfrac{1}{2}) $ is symmetric about $b + \tfrac{1}{2}, $
$ \pgt^\prime(b) = \tfrac{1}{2} $ in this case.

\qedhere
\end{proof}

\begin{proposition}[Worst-case pattern]\label{prop:worst}
For integer $\deltaz \geq 2$ (equivalently, for $ \az \leq \Soz - \Szo - 2 $),
the supremum of $\pgt(a, b)$ over patterns consistent with $\Hz: \Ao \leq \az$
occurs at the unique point
\begin{align}
  b^\star &= \min\bigl(\Szo + \Soo,\; \Szz + \Soz - \deltaz\bigr),
  & a^\star &= \min\bigl(\Szo + \Soo + \deltaz,\; \Szz + \Soz \bigr), \label{eqn:worst} \\
  \intertext{or, when we assume monotonic treatment effects,}
  b^\star &= \min\bigl(\Szo + \Soo,\; \Soz - \deltaz\bigr),
  & a^\star &= \min\bigl(\Szo + \Soo + \deltaz,\; \Soz \bigr). \label{eqn:worst-mon}
\end{align}

For $\deltaz \in \{0, 1\}$ (equivalently, for
$ \az \in \{ \Soz - \Szo, \Soz - \Szo - 1\} $), the supremum
equals $1$ and $\tfrac{1}{2}$ respectively; the test cannot reject at
conventional $\alpha$.
\end{proposition}
\begin{proof}
By Lemma~\ref{lem:mono-a}, the supremum on each horizontal slice lies
on the boundary $a = b + \deltaz$. By Lemma~\ref{lem:mono-b}, the
supremum on the boundary is attained at the largest $b$ admitted by
the box constraints, which is~(\ref{eqn:worst}),
or~(\ref{eqn:worst-mon}) when we assume monotonic treatment effects.
The boundary cases follow from the explicit values in
Lemma~\ref{lem:mono-b}.
\end{proof}

The two cases reflected in the min formulae correspond to whether the
diagonal line corresponding to the consistency constraint intersects
the top or right boundary of the box, resp. In
Figure~\ref{fig:consistency}, the line intersected the top of the box,
so $ b^\star = \Szo + \Soo $ and
$ a^\star = \Szo + \Soo + \deltaz = \Soz + \Soo - \az. $
This will be the case whenever $ \az \geq \Soo - \Szz $ (no
monotonicity assumption) or when $ \az \geq \Soo $ (monotonic
treatment effects). The practical consequence: testing $\Hz: \Ao \leq \az$
requires a \emph{single} Binomial tail calculation, computed at the
boundary corner $(a^\star, b^\star)$.

\paragraph{Example, continued.} Along the consistency boundary
$\mpo = \mpz - 40$, the box constraints permit
$\mpo \in [0, \min(\Szo + \Soo,\, \Szz + \Soz - 40)]
   = [0, \min(130,\, 830)] = [0, 130]$.
Lemma~\ref{lem:mono-b} (with $\deltaz = 40 \geq 2$) says
$\pgt(\mpo + 40, \mpo)$ is strictly increasing in $\mpo$ along this
range; Figure~\ref{fig:lemma2} plots the curve. By
Proposition~\ref{prop:worst} the worst pattern is
$(a^\star, b^\star) = (170, 130)$, with
$\pgt(a^\star, b^\star) \approx 0.0121 < 0.05$. We reject the
composite null $\Hz: \Ao \leq 0$ at $\alpha = 0.05$, on the strength
of this single Binomial tail calculation rather than the
$871 \times 131 = 114{,}101$ patterns enumerated
by~(\ref{eqn:eq-count}).

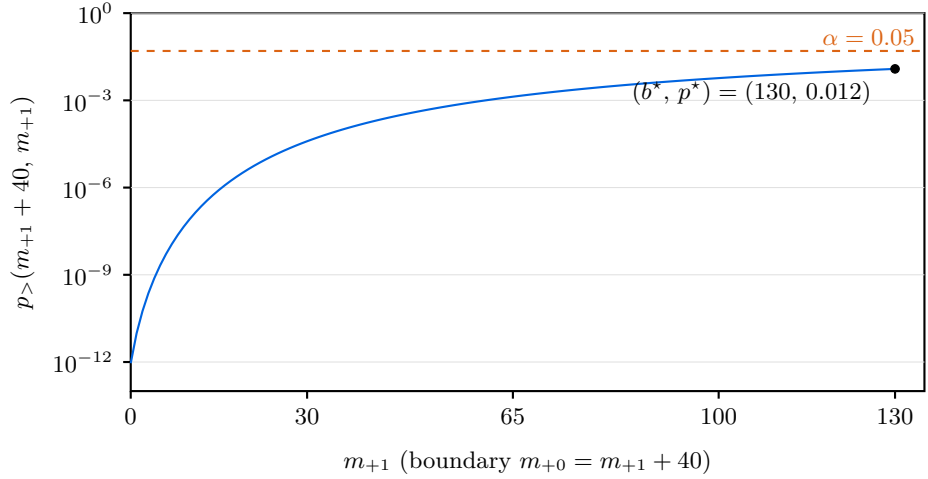
\begin{figure}[ht]
\centering
\begin{tikzpicture}
  \draw[thick] (0,0) rectangle (10.5,5.0);
  \draw[gray!25, very thin] (0,5.0000) -- (10.5,5.0000);
  \draw[thick] (0,5.0000) -- (-0.10,5.0000) node[anchor=east,font=\small] {$10^{0}$};
  \draw[gray!25, very thin] (0,3.8462) -- (10.5,3.8462);
  \draw[thick] (0,3.8462) -- (-0.10,3.8462) node[anchor=east,font=\small] {$10^{-3}$};
  \draw[gray!25, very thin] (0,2.6923) -- (10.5,2.6923);
  \draw[thick] (0,2.6923) -- (-0.10,2.6923) node[anchor=east,font=\small] {$10^{-6}$};
  \draw[gray!25, very thin] (0,1.5385) -- (10.5,1.5385);
  \draw[thick] (0,1.5385) -- (-0.10,1.5385) node[anchor=east,font=\small] {$10^{-9}$};
  \draw[gray!25, very thin] (0,0.3846) -- (10.5,0.3846);
  \draw[thick] (0,0.3846) -- (-0.10,0.3846) node[anchor=east,font=\small] {$10^{-12}$};
  \draw[thick] (0.0000,0) -- (0.0000,-0.10) node[anchor=north,font=\small] {$0$};
  \draw[thick] (2.3333,0) -- (2.3333,-0.10) node[anchor=north,font=\small] {$30$};
  \draw[thick] (5.0556,0) -- (5.0556,-0.10) node[anchor=north,font=\small] {$65$};
  \draw[thick] (7.7778,0) -- (7.7778,-0.10) node[anchor=north,font=\small] {$100$};
  \draw[thick] (10.1111,0) -- (10.1111,-0.10) node[anchor=north,font=\small] {$130$};
  \node[anchor=north,font=\small] at (5.25, -0.65) {$\mpo$ (boundary $\mpz = \mpo + 40$)};
  \node[anchor=south,rotate=90,font=\small] at (-1.15, 2.5) {$p_{>}(\mpo+40,\,\mpo)$};
  \draw[metaorange, dashed, thick] (0,4.4996) -- (10.5,4.4996);
  \node[metaorange, anchor=south east, font=\small, inner sep=2pt] at (10.45,4.4996) {$\alpha=0.05$};
  \draw[metablue, thick] (0.0000,0.3688) -- (0.0778,0.7655) -- (0.1556,1.0580) -- (0.2333,1.2938) -- (0.3111,1.4919) -- (0.3889,1.6628) -- (0.4667,1.8127) -- (0.5444,1.9458) -- (0.6222,2.0654) -- (0.7000,2.1736) -- (0.7778,2.2721) -- (0.8556,2.3624) -- (0.9333,2.4456) -- (1.0111,2.5225) -- (1.0889,2.5939) -- (1.1667,2.6604) -- (1.2444,2.7226) -- (1.3222,2.7808) -- (1.4000,2.8355) -- (1.4778,2.8870) -- (1.5556,2.9356) -- (1.6333,2.9815) -- (1.7111,3.0250) -- (1.7889,3.0663) -- (1.8667,3.1055) -- (1.9444,3.1428) -- (2.0222,3.1783) -- (2.1000,3.2122) -- (2.1778,3.2447) -- (2.2556,3.2756) -- (2.3333,3.3053) -- (2.4111,3.3338) -- (2.4889,3.3611) -- (2.5667,3.3873) -- (2.6444,3.4125) -- (2.7222,3.4367) -- (2.8000,3.4601) -- (2.8778,3.4826) -- (2.9556,3.5042) -- (3.0333,3.5252) -- (3.1111,3.5454) -- (3.1889,3.5649) -- (3.2667,3.5838) -- (3.3444,3.6021) -- (3.4222,3.6198) -- (3.5000,3.6369) -- (3.5778,3.6535) -- (3.6556,3.6696) -- (3.7333,3.6853) -- (3.8111,3.7005) -- (3.8889,3.7152) -- (3.9667,3.7295) -- (4.0444,3.7434) -- (4.1222,3.7570) -- (4.2000,3.7701) -- (4.2778,3.7830) -- (4.3556,3.7954) -- (4.4333,3.8076) -- (4.5111,3.8194) -- (4.5889,3.8310) -- (4.6667,3.8422) -- (4.7444,3.8532) -- (4.8222,3.8639) -- (4.9000,3.8744) -- (4.9778,3.8846) -- (5.0556,3.8946) -- (5.1333,3.9044) -- (5.2111,3.9139) -- (5.2889,3.9232) -- (5.3667,3.9323) -- (5.4444,3.9412) -- (5.5222,3.9499) -- (5.6000,3.9585) -- (5.6778,3.9668) -- (5.7556,3.9750) -- (5.8333,3.9830) -- (5.9111,3.9908) -- (5.9889,3.9985) -- (6.0667,4.0060) -- (6.1444,4.0134) -- (6.2222,4.0207) -- (6.3000,4.0278) -- (6.3778,4.0347) -- (6.4556,4.0415) -- (6.5333,4.0483) -- (6.6111,4.0548) -- (6.6889,4.0613) -- (6.7667,4.0676) -- (6.8444,4.0739) -- (6.9222,4.0800) -- (7.0000,4.0860) -- (7.0778,4.0919) -- (7.1556,4.0977) -- (7.2333,4.1034) -- (7.3111,4.1090) -- (7.3889,4.1145) -- (7.4667,4.1199) -- (7.5444,4.1252) -- (7.6222,4.1305) -- (7.7000,4.1356) -- (7.7778,4.1407) -- (7.8556,4.1457) -- (7.9333,4.1506) -- (8.0111,4.1554) -- (8.0889,4.1602) -- (8.1667,4.1649) -- (8.2444,4.1695) -- (8.3222,4.1740) -- (8.4000,4.1785) -- (8.4778,4.1829) -- (8.5556,4.1873) -- (8.6333,4.1916) -- (8.7111,4.1958) -- (8.7889,4.1999) -- (8.8667,4.2040) -- (8.9444,4.2081) -- (9.0222,4.2121) -- (9.1000,4.2160) -- (9.1778,4.2199) -- (9.2556,4.2237) -- (9.3333,4.2275) -- (9.4111,4.2312) -- (9.4889,4.2348) -- (9.5667,4.2385) -- (9.6444,4.2420) -- (9.7222,4.2455) -- (9.8000,4.2490) -- (9.8778,4.2525) -- (9.9556,4.2558) -- (10.0333,4.2592) -- (10.1111,4.2625);
  \filldraw[black] (10.1111,4.2625) circle (1.6pt);
  \node[anchor=north east, font=\small, inner sep=3pt] at (9.9111,4.2125) {$(b^\star,\,p^\star)=(130,\,0.012)$};
\end{tikzpicture}
\caption{Lemma~\ref{lem:mono-b} in the running example: along the
consistency boundary $\mpo = \mpz - 40$, the worst-case p-value is
strictly increasing in $\mpo$. The supremum is attained at the
right endpoint $b^\star = \min(130, 830) = 130$, with
$a^\star = 170$ and
$\pgt(170, 130) \approx 0.0121$. Since this falls below
$\alpha = 0.05$ (orange dashed line), $\Hz: \Ao \leq 0$ is rejected.}
\label{fig:lemma2}
\end{figure}

\paragraph{A test for the upper bound.}
We handle the hypothesis $\Hz^\prime: \Ao \geq \az$ symmetrically.
For the upper-bound regime $\az \geq \Soz - \Szo$, consistency
becomes $\mpoo - \mpzo \geq -\deltaz$: the consistency region lies
\emph{above} the diagonal line but subject to the same box constraints
illustrated in Figure~\ref{fig:consistency}. For evaluating a
particular pattern, the natural test is left-tailed,
$\plt(a, b) := \Pr\{\Binom(a + b, \tfrac{1}{2}) \leq a\}$. The
analogue of Lemma~\ref{lem:mono-a} is that for fixed $ a, $ the
supremum lies at the bottom of the consistency region, along the
diagonal line. The analogue of Lemma~\ref{lem:mono-b} demonstrates the
supremum lies on the right side of this diagonal (when $ \deltaz \leq -2 $).
The worst-case pattern is given by the same formula as
in~(\ref{eqn:worst}) and~(\ref{eqn:worst-mon}).

\paragraph{Inference for $\Az$.}
The same approach works for the net effect among the controls, $\Az$,
with treatment and control exchanged. Inference about $\Ao$ rests on
the \emph{control} potential outcomes implied by a sharp null,
$r_{Csi}^0 = R_{si} - Z_{si}\,\tau_{si}^0$, and runs McNemar's test on
the resulting control-based discordant pairs. Inference about $\Az$
rests on the \emph{treated} potential outcomes implied by the same
sharp null, $r_{Tsi}^0 = R_{si} + (1 - Z_{si})\,\tau_{si}^0$, and runs
McNemar on the treatment-based discordant pairs.

We still categorize the pairs by their observed outcomes
$(R_{si_T}, R_{si_C})$, so the counts $\Sjk$ are unchanged. What
differs is the hypothesized effect rendering a pair concordant or
discordant, which now acts on the \emph{control} unit. The analogue of
Table~\ref{tbl:treatment-concordant} is obtained by negating its
Discordant column: rows $(0,0)$ and $(1,1)$ flip sign, while the
Concordant column is unchanged. The test statistic is still
the number of discordant pairs in which the treated unit is the
success-under-treatment, which now has formula $\mop := \moz + \moo$.
Writing $\mzp := \mzz + \mzo$, the realized net effect is
$\Az = \Soz - \Szo + \mzp - \mop$, analogous to
the identity for $\Ao$. The box constraints of \S\ref{sec:hardest}
then become
\begin{equation}\label{eqn:a0-box}
   0 \leq \mzp \leq \Szz + \Szo,
   \qquad
   0 \leq \mop \leq \Soz + \Soo,
\end{equation}
with $\mop - \mzp \geq \Soz - \Szo - \az =: \deltaz $
completing the consistency region. The analogues of
Lemmas~\ref{lem:mono-a} and~\ref{lem:mono-b} set $a = \mop$ and
$b = \mzp$ then proceed verbatim. Table~\ref{tbl:worst-case}
summarizes the worst-case patterns for each scenario.

\begin{table}[ht]
\centering
\small
\caption{Worst-case patterns $(a^\star, b^\star)$ for the matched-pair
tests, by estimand and monotonicity assumption; the worst-case pattern
is the same for both test directions. When testing
$\Hz: A_{0/1} \leq \az$ vs.\ $\Ho: A_{0/1} > \az$ (for
$\az \leq \Soz - \Szo - 2$), the worst-case $p$-value is
$\Pr\{\Binom(a^\star + b^\star, \tfrac{1}{2}) \geq a^\star\}$; when
testing $\Hz^\prime: A_{0/1} \geq \az$ vs.\ $\Ho^\prime: A_{0/1} < \az$
(for $\az \geq \Soz - \Szo + 2$), it is
$\Pr\{\Binom(a^\star + b^\star, \tfrac{1}{2}) \leq a^\star\}$. Within
each $\min$, the first argument corresponds to the consistency line
meeting the top of the constraint box, the second to its right edge.}
\label{tbl:worst-case}
\begin{tabular}{c c c c}
\toprule
Estimand & Monotonic Treatment Effects? & $a^\star$ & $b^\star$ \\
\midrule
$\Ao$ & N & $\min(\Soz+\Soo-\az,\; \Szz+\Soz)$
      & $\min(\Szo+\Soo,\; \Szz+\Szo+\az)$ \\
$\Ao$ & Y & $\min(\Soz+\Soo-\az,\; \Soz)$
      & $\min(\Szo+\Soo,\; \Szo+\az)$ \\
$\Az$ & N & $\min(\Szz+\Soz-\az,\; \Soz+\Soo)$
      & $\min(\Szz+\Szo,\; \Szo+\Soo+\az)$ \\
$\Az$ & Y & $\min(\Szz+\Soz-\az,\; \Soz)$
      & $\min(\Szz+\Szo,\; \Szo+\az)$ \\
\bottomrule
\end{tabular}
\end{table}

\section{Prediction sets via binary search}\label{sec:practical}
Inverting the tests gives prediction sets for $\Ao$. When $ \az $ is
at its lowest possible value, $ -f_T, $ then
$ \deltaz = \Szz + \Soz, $ $ b^\star = 0, $ and
\begin{align}
  \pgt^\prime(b^\star) &= \Pr\{ \Binom(\Szz + \Soz, \tfrac{1}{2}) \geq \Szz + \Soz \} \nonumber \\
  &= 2^{-(\Szz + \Soz)} = 2^{-f_C}. \nonumber
\end{align}
If there are at least 5 failures among the control units, we can
reject $\Hz: \az =−f_T$ at the conventional $\alpha = 0.05$ level.

As $\az$ grows, $\deltaz$ shrinks, and the worst-case pattern moves
first up and then left. Both lead to larger p-values by arguments
similar to those underlying Lemmas~\ref{lem:mono-a}
and~\ref{lem:mono-b}. When $ \az = \Soz - \Szo - 1, $ $ \deltaz = 1, $
$ \pgt^\prime(b^\star) = \tfrac{1}{2}, $ and we cannot reject. The
lower endpoint of the one-sided $1 - \alpha$ prediction set,
\begin{displaymath}
  \Lalpha := \inf\{\az \in \mathbb{Z}
                \mid \pgt(a^\star(\az), b^\star(\az)) > \alpha\},
\end{displaymath}
may therefore be located by binary search over the integer range
$\az \in [-f_T, \Soz - \Szo - 1]$ in $O(\log S)$ Binomial tail
calculations. The other one-sided test gives an upper endpoint
$\Ualpha$. A two-sided $1 - 2\alpha$ prediction set for $\Ao$ is
$[\Lalpha, \Ualpha]$.

\paragraph{Example, continued.}
Sweeping $\az$ in the running example pins both endpoints. Throughout
the relevant range, the worst $b^\star$ stays at the box-constraint
ceiling $\Szo + \Soo = 130$, and $a^\star = 130 + \deltaz$. For
the lower endpoint with $\alpha = 0.05$,
\begin{align}
   \pgt(159, 130) &\approx 0.0497 & (\az &= 11, \ \deltaz = 29), \nonumber \\
   \pgt(158, 130) &\approx 0.0557 & (\az &= 12, \ \deltaz = 28). \nonumber
\end{align}
The worst-case p-value crosses $0.05$ between $\az = 11$ and $\az = 12$,
so $L_{0.05} = 12$. The upper bound test gives
\begin{align}
   \plt(104, 130) &\approx 0.0510 & (\az &= 66, \ \deltaz = -26), \nonumber \\
   \plt(103, 130) &\approx 0.0441 & (\az &= 67, \ \deltaz = -27), \nonumber
\end{align}
so $U_{0.05} = 66$. The one-sided 95\% prediction sets are
$\{\Ao \geq 12\}$ and $\{\Ao \leq 66\}$ respectively; their
intersection is the two-sided 90\% prediction set
$\Ao \in [12, 66]$.

\paragraph{The no-information band.}
The boundary cases of Proposition~\ref{prop:worst} are interpretable.
When $\deltaz = 0$ (i.e., $\az = \Soz - \Szo$, the observed
treated-success surplus), the consistent pattern $(0, 0)$ corresponds
to ``no treatment effect on any treated unit,'' yields an empty
McNemar test, and has p-value $1$. When $\deltaz = 1$, every
boundary pattern has p-value exactly $\tfrac{1}{2}$. When
$\deltaz < 0$ (i.e., $\az > \Soz - \Szo$), the consistency
constraint $\mpz \geq \mpo + \deltaz$ is weaker than the natural
non-negativity bound $\mpz, \mpo \geq 0$; the unconstrained worst-case
sits at $\mpz < \mpo$, where the right-tailed p-value
$\Pr\{\Binom(\mpz + \mpo, \tfrac{1}{2}) \geq \mpz\}$ approaches $1$.
In all three cases the test cannot reject at conventional $\alpha$,
so the lower endpoint of any one-sided prediction set for $\Ao$ lies
in $\deltaz \geq 2$, i.e., the prediction set necessarily contains
$\{\az \geq \Soz - \Szo - 1\}$.

\paragraph{The Hodges-Lehmann point estimate.}
The same machinery yields a closed-form point estimate. The
worst-case statistic $a^\star(\az)$ has null mean
$\tfrac{1}{2}(a^\star + b^\star)(\az)$, so the
estimating function $S(\az) := a^\star(\az) - \tfrac{1}{2}(a^\star
+ b^\star)(\az)$ has null expectation zero. Solving $S(\az) = 0$
for $\az$ defines the point estimator $\hat A$, an instance of the
method-of-moments procedure analyzed by
\citet[Lemma~1.2]{maritz-1995-distribution-free} and consistent
under the regularity conditions stated there. This is the
construction \citet{hodges-1963-point-estimates} introduced for
inverting a location-shift rank test; here we apply it to matched
pairs with binary outcomes. Equating $a^\star$ to its null mean forces
$a^\star = b^\star$, hence $\deltaz = 0$ via
Proposition~\ref{prop:worst} (with or without assuming monotonic
treatment effects), and so the point estimate is
\begin{equation}\label{eqn:hl-a1}
   \hat A = \Soz - \Szo.
\end{equation}
The worst-case pattern $(a^\star, b^\star)$ is the same for $ \Hz $ or
$ \Hz^\prime $. They differ for inferences regarding $ \Ao $ or $ \Az $,
but for $ \Az $, too, the strategy leads to Equation~(\ref{eqn:hl-a1})
as the point estimate. Equation~(\ref{eqn:hl-a1}) is the difference in
observed discordant-pair counts, which matches classic discussions
based on sampling from a population. Our derivation shows the
assignment mechanism (random in an experiment or strongly ignorable in
an observational study) justifies the point estimate, without assuming
the units are samples from a population.

\paragraph{Example, continued.}
The Hodges-Lehmann point estimate for the net treatment effect among
treated units is $\hat A = \Soz - \Szo = 40$. The prediction set
above brackets the point estimate asymmetrically, with $28$ units of
room below and $26$ above. The asymmetry reflects the larger Binomial
variance (proportional to $n = a^\star + b^\star$) at the lower-bound
crossing ($n = 288$) than at the upper ($n = 234$), so we need a
slightly wider half-interval below the point estimate than above.

\paragraph{Large-sample approximation.}
For large $S$ the binary search admits a closed-form Gaussian
approximation. Standardize the worst-case McNemar statistic against
its $\Binom(a^\star + b^\star, \tfrac{1}{2})$ null, with a continuity
correction $a^\star \mapsto a^\star - \tfrac{1}{2}$ for the integer
support so we approximate
$\Pr\{a^\star \geq a\} = \Pr\{a^\star > a - 1\}$ symmetrically,
\begin{displaymath}
   Z := \frac{a^\star - \tfrac{1}{2} - \tfrac{1}{2}(a^\star + b^\star)}
             {\tfrac{1}{2}\sqrt{a^\star + b^\star}}
       = \frac{\deltaz - 1}{\sqrt{2 b^\star + \deltaz}}.
\end{displaymath}
Setting $Z = \zalpha$ and squaring yields a quadratic in
$\deltaz$, but $b^\star = \min(\Szo + \Soo,\, \Szz + \Soz - \deltaz)$
itself depends on $\deltaz$, so the positive root takes one of two
forms according to which face of the constraint box the consistency
line meets (\S\ref{sec:hardest}; Figure~\ref{fig:consistency}).

\emph{Top face} ($b^\star = \Szo + \Soo$, when $\az \geq \Soo - \Szz$;
equivalently, when $ \deltaz \leq \Szz + \Soz - \Szo - \Soo $):
$\deltaz^2 - (2 + \zalpha^2) \deltaz + 1 - 2(\Szo + \Soo) \zalpha^2 = 0$, so
\begin{equation}\label{eqn:root-top}
   \deltaz^\star = 1 + \tfrac{1}{2} \zalpha^2
     + \sqrt{2 \zalpha^2 (\Szo + \Soo)
             + \zalpha^2 + \tfrac{1}{4} \zalpha^4}.
\end{equation}

\emph{Right face} ($b^\star = \Szz + \Soz - \deltaz$, hence
$a^\star = \Szz + \Soz$, when $\az \leq \Soo - \Szz$): here
$\deltaz^2 - (2 - \zalpha^2) \deltaz + 1 - 2(\Szz + \Soz)
\zalpha^2 = 0$, so
\begin{equation}\label{eqn:root-right}
   \deltaz^\star = 1 - \tfrac{1}{2} \zalpha^2
     + \sqrt{2 \zalpha^2 (\Szz + \Soz)
             - \zalpha^2 + \tfrac{1}{4} \zalpha^4}.
\end{equation}

The lower endpoint is $\Lalpha = (\Soz - \Szo)
- \lfloor \deltaz^\star \rfloor$, taking whichever root is
self-consistent (meaning, the one whose implied
$\az = (\Soz - \Szo) - \lfloor \deltaz^\star \rfloor$
is on the corresponding side of $ \Soo - \Szz $).
Since $Z$ is continuous and strictly increasing in
$\deltaz$, with only a change of slope where the faces meet
($\deltaz = \Szz + \Soz - \Szo - \Soo$), exactly one root is
self-consistent; on the rare occasion that the flooring straddles the
crossover, $\Lalpha$ is settled by evaluating $Z$ directly on the few
integers between the two candidates. The upper endpoint $\Ualpha$
follows from the left-tailed test, its root carrying
$-\tfrac{1}{2} \zalpha^2$ in place of
$+\tfrac{1}{2} \zalpha^2$; the leading $+1$ from the continuity
correction is common to both endpoints.

Factoring $\zalpha$ out of the square root in
(\ref{eqn:root-top})--(\ref{eqn:root-right}) gives, on each face,
$\deltaz^\star = 1 \pm \tfrac{1}{2} \zalpha^2
+ \zalpha\sqrt{2 \Sbullet \pm 1 + \tfrac{1}{4} \zalpha^2}$,
where the ceiling $\Sbullet$ is $\Szo + \Soo$ with the upper signs
(top face) and $\Szz + \Soz$ with the lower (right face). In the
typical regime a single face -- the one with the smaller ceiling
$\min(\Szz + \Soz,\, \Szo + \Soo)$ -- binds both endpoints;
dropping the $\pm 1$ inside the radical (negligible against
$2 \Sbullet$), the prediction-set endpoints are
\begin{equation}\label{eqn:LU-approx}
   \Lalpha,\; \Ualpha \;\approx\;
   (\Soz - \Szo) - \tfrac{1}{2} \zalpha^2
   \,\mp\, \left(1 + \zalpha
     \sqrt{2 \min(\Szz + \Soz,\, \Szo + \Soo)
           + \tfrac{1}{4} \zalpha^2}\right).
\end{equation}
The interval is nearly symmetric about the point
estimate~(\ref{eqn:hl-a1}). The $-\tfrac{1}{2} \zalpha^2$ term
shifts both endpoints slightly downward, so the lower half-width
carries a slight excess over the upper, a closed-form expression of
the asymmetry observed numerically above.

For the running example the binding ceiling is
$\min(\Szz+\Soz,\, \Szo+\Soo) = \Szo + \Soo = 130$.
The lower root~(\ref{eqn:root-top}) is
$\deltaz^\star \approx 28.96$, giving
$L_{0.05} = 40 - \lfloor 28.96 \rfloor = 12$; the mirror root is
$\approx 26.16$, giving $U_{0.05} = 40 + \lfloor 26.16 \rfloor = 66$.
The Gaussian approximation lands on the binary-search values
$[12, 66]$. The approximation requires $\deltaz \gg 2$ and should not
be applied inside the no-information band.

\paragraph{Inference for $\Az$.}\label{par:a0}
The topics of this section carry over to $\Az$ unchanged. With the
worst-case patterns already established in \S\ref{sec:hardest}
(Table~\ref{tbl:worst-case}, under the swapped box
constraints~(\ref{eqn:a0-box})), the binary-search inversion of the
exact test, its large-sample Gaussian approximation, and the
Hodges-Lehmann point estimate all apply verbatim, differing only in
which marginal of the $\Sjk$ table supplies the box ceilings.
For the running example, an exact two-sided 90\% prediction set for
$ \Az $ is $[11, \, 72]$. The Gaussian approximation leads to the same
interval.

Table~\ref{tbl:large-sample} collects the resulting large-sample
endpoints for both one-sided tests, for $ \Ao $ and $ \Az $, with and
without the monotonicity assumption, and for each of the two
constraint-box faces the consistency line may meet.

\begin{table}[tp]
\centering
\small
\setlength{\tabcolsep}{4pt}
\caption{Large-sample Gaussian approximations to the prediction-set
endpoints, by estimand ($\Ao$ or $\Az$), monotonicity assumption,
endpoint (lower $\Lalpha$ or upper $\Ualpha$), and which face of the
constraint box the consistency line meets --- the top ($T$) or right
($R$). Throughout, $z := \zalpha$, and each expression is the
continuity-corrected root of the standardization in
\S\ref{sec:practical}. For a given estimand and monotonicity
assumption, the applicable face is the one whose endpoint satisfies
the criterion in the final column; the integer endpoints are then
$\Lalpha = \lceil \cdot \rceil$ and $\Ualpha = \lfloor \cdot
\rfloor$ of the tabulated value.}
\label{tbl:large-sample}
\begin{tabular}{c c l c}
\toprule
Estimand & \shortstack{Monotonic\\ Treatment\\ Effects?} & Formula & Criterion \\
\midrule
\multirow{8}{*}{$\Ao$}
 & \multirow{4}{*}{N}
   & $\Lalpha^T = \Soz-\Szo - 1 - \tfrac{1}{2}z^2 - z\sqrt{2(\Szo+\Soo) + 1 + \tfrac{1}{4}z^2}$ & $\Lalpha^T \geq \Soo-\Szz$ \\
 & & $\Lalpha^R = \Soz-\Szo - 1 + \tfrac{1}{2}z^2 - z\sqrt{2(\Szz+\Soz) - 1 + \tfrac{1}{4}z^2}$ & $\Lalpha^R \leq \Soo-\Szz$ \\
 & & $\Ualpha^T = \Soz-\Szo + 1 - \tfrac{1}{2}z^2 + z\sqrt{2(\Szo+\Soo) - 1 + \tfrac{1}{4}z^2}$ & $\Ualpha^T \geq \Soo-\Szz$ \\
 & & $\Ualpha^R = \Soz-\Szo + 1 + \tfrac{1}{2}z^2 + z\sqrt{2(\Szz+\Soz) + 1 + \tfrac{1}{4}z^2}$ & $\Ualpha^R \leq \Soo-\Szz$ \\
\cmidrule(l){2-4}
 & \multirow{4}{*}{Y}
   & $\Lalpha^T = \Soz-\Szo - 1 - \tfrac{1}{2}z^2 - z\sqrt{2(\Szo+\Soo) + 1 + \tfrac{1}{4}z^2}$ & $\Lalpha^T \geq \Soo$ \\
 & & $\Lalpha^R = \Soz-\Szo - 1 + \tfrac{1}{2}z^2 - z\sqrt{2 \Soz - 1 + \tfrac{1}{4}z^2}$ & $\Lalpha^R \leq \Soo$ \\
 & & $\Ualpha^T = \Soz-\Szo + 1 - \tfrac{1}{2}z^2 + z\sqrt{2(\Szo+\Soo) - 1 + \tfrac{1}{4}z^2}$ & $\Ualpha^T \geq \Soo$ \\
 & & $\Ualpha^R = \Soz-\Szo + 1 + \tfrac{1}{2}z^2 + z\sqrt{2 \Soz + 1 + \tfrac{1}{4}z^2}$ & $\Ualpha^R \leq \Soo$ \\
\midrule
\multirow{8}{*}{$\Az$}
 & \multirow{4}{*}{N}
   & $\Lalpha^T = \Soz-\Szo - 1 - \tfrac{1}{2}z^2 - z\sqrt{2(\Szz+\Szo) + 1 + \tfrac{1}{4}z^2}$ & $\Lalpha^T \geq \Szz-\Soo$ \\
 & & $\Lalpha^R = \Soz-\Szo - 1 + \tfrac{1}{2}z^2 - z\sqrt{2(\Soz+\Soo) - 1 + \tfrac{1}{4}z^2}$ & $\Lalpha^R \leq \Szz-\Soo$ \\
 & & $\Ualpha^T = \Soz-\Szo + 1 - \tfrac{1}{2}z^2 + z\sqrt{2(\Szz+\Szo) - 1 + \tfrac{1}{4}z^2}$ & $\Ualpha^T \geq \Szz-\Soo$ \\
 & & $\Ualpha^R = \Soz-\Szo + 1 + \tfrac{1}{2}z^2 + z\sqrt{2(\Soz+\Soo) + 1 + \tfrac{1}{4}z^2}$ & $\Ualpha^R \leq \Szz-\Soo$ \\
\cmidrule(l){2-4}
 & \multirow{4}{*}{Y}
   & $\Lalpha^T = \Soz-\Szo - 1 - \tfrac{1}{2}z^2 - z\sqrt{2(\Szz+\Szo) + 1 + \tfrac{1}{4}z^2}$ & $\Lalpha^T \geq \Szz$ \\
 & & $\Lalpha^R = \Soz-\Szo - 1 + \tfrac{1}{2}z^2 - z\sqrt{2 \Soz - 1 + \tfrac{1}{4}z^2}$ & $\Lalpha^R \leq \Szz$ \\
 & & $\Ualpha^T = \Soz-\Szo + 1 - \tfrac{1}{2}z^2 + z\sqrt{2(\Szz+\Szo) - 1 + \tfrac{1}{4}z^2}$ & $\Ualpha^T \geq \Szz$ \\
 & & $\Ualpha^R = \Soz-\Szo + 1 + \tfrac{1}{2}z^2 + z\sqrt{2 \Soz + 1 + \tfrac{1}{4}z^2}$ & $\Ualpha^R \leq \Szz$ \\
\bottomrule
\end{tabular}
\end{table}

\section{Confidence sets for the ATE}\label{sec:ate}
Section~\ref{sec:practical} produces prediction sets for both
$\Ao(Z, \tau)$ and $\Az(Z, \tau)$. Combining the two yields a
confidence set for the $\ATE$ via the
constraint~(\ref{eqn:sum-constraint}).

Proposition~1 of \citet{rigdon-2015-attributable-effect} governs the
combination: if $\{L_1, L_1 + 1, \ldots, U_1\}$ is a $1 - \alpha/2$
prediction set for $\Ao$ and $\{L_0, L_0 + 1, \ldots, U_0\}$ is a
$1 - \alpha/2$ prediction set for $\Az$, then
\begin{displaymath}
  \left\{ \frac{L_1 + L_0}{2S},\;
          \frac{L_1 + L_0 + 1}{2S},\;
          \ldots,\;
          \frac{U_1 + U_0}{2S} \right\}
\end{displaymath}
is a $1 - \alpha$ confidence set for the $\ATE$. The proof is a
Bonferroni argument exploiting~(\ref{eqn:sum-constraint}): the joint
coverage of two $1 - \alpha/2$ prediction sets is at least
$1 - \alpha$. The linear constraint pins the combined endpoints
to the extremes of the two prediction sets.

Combining two $1 - \alpha$ prediction sets without the half-$\alpha$
correction does not guarantee $1 - \alpha$ coverage, which is
addressed by the Bonferroni adjustment. Rigdon and Hudgens exhibit a
small finite-population example ($m = 4$ of $n = 9$ treated) in which
the naive combination achieves only $92\%$ coverage where $95\%$ was
nominally claimed.

\paragraph{End-to-end procedure.}
\begin{enumerate}
  \item Build the two-sided $1 - \alpha/2$ prediction set
    $[L_1, U_1]$ for $\Ao$ by binary search at level $\alpha/4$ on
    each side.
  \item Build the two-sided $1 - \alpha/2$ prediction set
    $[L_0, U_0]$ for $\Az$ by the mirror procedure at level $\alpha/4$
    on each side.
  \item Form the $\ATE$ confidence set
    $\bigl[(L_1 + L_0)/(2S),\, (U_1 + U_0)/(2S)\bigr]$.
\end{enumerate}
The total computational cost is four binary searches, i.e., $O(\log S)$ Binomial
tail calculations.

\paragraph{Example, continued.}
For a $1 - \alpha = 0.90$ confidence set on the $\ATE$ we need
$1 - \alpha/2 = 0.95$ two-sided prediction sets for $\Ao$ and $\Az$.
These are the intersection of two one-sided intervals at level
$\alpha/4 = 0.025$, half the per-side level used in
\S\ref{sec:practical}'s example. Applying that procedure at the
stricter level (for $\Az$, under the swapped box
ceilings~(\ref{eqn:a0-box})) gives $\Ao \in [6, 70]$ and
$\Az \in [5, 79]$. Combining,
\begin{displaymath}
   \ATE \in
   \left[\frac{L_1 + L_0}{2S},\, \frac{U_1 + U_0}{2S}\right]
   = \left[\frac{11}{2000},\, \frac{149}{2000}\right]
   = [0.0055,\, 0.0745]
\end{displaymath}
is a 90\% confidence set for the $\ATE$. It brackets the
randomization-based point estimate
$\widehat{\ATE} = (\Soz - \Szo)/S = 0.040$ symmetrically, with
half-width $0.0345$ on each side. The symmetry is structural: the
$\Ao$ prediction set $[6, 70]$ leans low about its pivot ($34$ below,
$30$ above) and the $\Az$ prediction set $[5, 79]$ leans high
($35$ below, $39$ above). The asymmetries cancel under the sum
$L_1 + L_0$ versus $U_1 + U_0$, as the swapped boxes guarantee.

\paragraph{Large-sample $\ATE$ confidence interval.}
Combining the Gaussian approximations of \S\ref{sec:practical}
(equation~(\ref{eqn:LU-approx}) and its $\Az$ mirror)
via the Bonferroni proposition above yields the closed-form $\ATE$
confidence interval
\begin{align}
   \widehat{\ATE} \,\pm\, \frac{\zalphaq}{2S}
   \Bigl(\sqrt{2 \min(\Soz+\Szz,\, \Szo+\Soo) + \tfrac{1}{4} \zalphaq^2} \hspace{20pt} \nonumber \\
       + \sqrt{2 \min(\Soz+\Soo,\, \Szo+\Szz) + \tfrac{1}{4} \zalphaq^2}\Bigr), \label{eqn:lsate}
\end{align}
where $\widehat{\ATE} = (\Soz - \Szo)/S$ and $\zalphaq$ is
the per-side critical value after the $\alpha \to \alpha/4$ Bonferroni
adjustment (one factor of two for two-sided, one for the $\Ao$/$\Az$
combination). For the running example the binding ceilings are
$\min(\Soz+\Szz, \Szo+\Soo) = 130$ and
$\min(\Soz+\Soo, \Szo+\Szz) = 170$,
and $\znf = 1.96$, giving a half-width of
$(1.96/2000) \bigl(\sqrt{2 \cdot 130 + \tfrac{1}{4}(1.96)^2}
+ \sqrt{2 \cdot 170 + \tfrac{1}{4}(1.96)^2}\bigr) \approx 0.0339$
at $1 - \alpha = 0.90$, in close agreement with the binary-search
half-width $0.0345$ above. The continuity corrections of
\S\ref{sec:practical} add an $O(1/S)$ widening to each combined
endpoint and are omitted here.

\paragraph{Point estimate of the ATE.}
Since $ \ATE = (\Ao + \Az) / (2 S) $ by Equation~(\ref{eqn:sum-constraint}),
and since the point estimates for $ \Ao $ and $ \Az $ are both
$ \Soz - \Szo, $ the point estimate for $ \ATE $ is
\begin{equation}\label{eqn:hl-ate}
   \widehat{\ATE} = \frac{\hat A}{S} = \frac{\Soz - \Szo}{S}.
\end{equation}
Again, this is the standard estimator, but here motivated by random
assignment rather than sampling from a population.

\paragraph{Comparison with the textbook McNemar interval.}
The standard McNemar treatment, e.g.\ \citep{fleiss-2013-proportions},
models the four cell counts $\{\Sjk\}$ as a multinomial sample from
a population of paired binary outcomes and targets the difference of
marginal success probabilities $\Delta := p_{1+} - p_{+1}$.
Its point estimate is $\hat\Delta = (\Soz - \Szo)/S$, identical to
$\widehat{\ATE}$ of~(\ref{eqn:hl-ate}), and its $1 - \alpha$
confidence interval is
$\hat\Delta \pm \frac{\zalphah}{S}\sqrt{\Szo + \Soz - (\Soz - \Szo)^2 / S}$,
with half-width scaling with the discordant-pair count. The two
procedures share the point estimate but rest on different foundations:
\begin{itemize}
\item \emph{Inferential target and coverage.} McNemar targets the
  marginal probabilities of a population, real or imagined, with
  coverage valid in repeated multinomial sampling and contingent on
  the population model. The present procedure targets the
  finite-population $\ATE$ for the $2S$ units actually observed, with
  coverage valid in repeated within-pair coin flips and contingent
  only on the pair-level assignment being random.
\item \emph{Randomization inference analyzes potential outcomes.} A
  sharp null fixes the unit-level effects $\tau_{si}^0$ and so imputes
  the missing control outcomes for the treated units. When asserting a
  net effect of $\az$ (and assuming monotonic treatment effects), the
  worst-case pattern posits $\az$ treated successes whose
  counterfactual was a failure. Each such imputation turns a would-be
  concordant pair discordant: the reference distribution is built on
  the \emph{null-implied} discordant count $\Szo + \Soz + \az$.
  The population null instead constrains only the marginal parameter
  $\Delta$, and because $\sum\sb{jk} p_{jk} = 1$ pins the discordant
  total to the observed concordant cells, it can \emph{rearrange}
  discordant pairs but never create them; its variance stays near the
  observed $\Szo + \Soz$ for every $\Deltaz$. The two coincide at
  $\az = 0$ and diverge toward the endpoints. This is a design-based
  feature, not a choice of variance estimator.
\item \emph{The role of monotonicity.} The count $\Szo + \Soz + \az$
  above assumes monotonic treatment effects ($\tau_{si} \geq 0$).
  Relaxing monotonicity lets a prevented success live in a concordant
  pair, so the worst-case denominator inflates further, to
  $2 b^\star + \deltaz$. In our example, this roughly doubles the
  binding box ceiling ($300 - \az$ versus $100 + \az$).
\item \emph{Bonferroni adjustment.} Combining the prediction sets for
  $\Ao$ and $\Az$ requires using $\zalphaq$ instead of
  $\zalphah,$ inflating confidence widths by about $19\%$ in
  this example. Yet using $\zalphah$ in~(\ref{eqn:lsate})
  sabotages the coverage of our interval.
\item \emph{Width.} The three effects multiply. At a common $90\%$
  level, the null acting on potential outcomes inflates the half-width
  by about $20\%$, relaxing monotonicity by a further $\approx 46\%$,
  and the Bonferroni upgrade by another $\approx 19\%$, so our
  half-width ($0.0339$) is about $1.20 \times 1.46 \times 1.19 \approx 2.1$
  times as wide as the textbook McNemar's ($1.645\sqrt{98.4}/1000 \approx 0.0163$).
\end{itemize}

The trade-off is the standard design-based vs.\ model-based dichotomy:
inference about causal effects vs.\ population parameters.

\section{Sensitivity analysis}\label{sec:sensitivity}
The procedure extends to observational studies via the sensitivity
model of \citet{rosenbaum-1987-sensitivity-analysis-matched-pairs}.
For each pair, let $\pis := \Pr\{Z_{s1} = 1 \mid Z_{s1} + Z_{s2} = 1\}$
be the probability that, given one unit of the pair is treated, it is
unit $1$. In the population before matching $\Pr\{Z_{si} = 1\}$ may
not be anything close to $\tfrac{1}{2}$; but if the two units of a
pair are equally likely to select treatment, then $\pis = \tfrac{1}{2}$,
recovering the within-pair coin flip
\citep[\S3.2]{rosenbaum-2020-design-of-observational-studies}. Under
unmeasured confounding bounded by $\Gamma \geq 1$, the standard
matched-pair sensitivity model assumes
\begin{equation}\label{eqn:sensitivity}
   \frac{1}{\Gamma + 1} \leq \pis \leq \frac{\Gamma}{\Gamma + 1}
   \qquad \text{for all } s.
\end{equation}
Under~(\ref{eqn:sensitivity}) the McNemar statistic is no longer a
binomial under the null hypothesis, but rather a sum of independent
Bernoulli variables whose success probabilities $\pis$ are unknown
but assumed to lie within the interval~(\ref{eqn:sensitivity}).

While the test statistic has an unknown distribution under the null
hypothesis, it is stochastically bounded by Binomial distributions
with probabilities of success corresponding to the bounds
in~(\ref{eqn:sensitivity}). For example, when testing a pattern
$ (\mpzo, \mpoo) $ consistent with $ \Hz: \Ao \leq \az $, the p-value
is bounded above by $ \Pr\left\{\Binom\left(\mpzo + \mpoo, \frac{\Gamma}{\Gamma+1}\right) \geq \mpzo\right\} $.
If the supremum of this quantity (over all consistent patterns) is
less than the significance threshold, we know the supremum of the
(unknown) p-values is also less than the threshold, and we reject
$ \Hz $ despite not knowing the actual p-value.

Lemmas~\ref{lem:mono-a} and~\ref{lem:mono-b} generalize and lead to
the same worst-case pattern $ (a^\star, b^\star) $ as summarized in
Table~\ref{tbl:worst-case}. (The $\tfrac{1}{4} \cdot \tfrac{\deltaz - 1}{b + \deltaz}$
in Lemma~\ref{lem:mono-b} becomes
$\tfrac{\Gamma}{(\Gamma + 1)^2} \cdot \tfrac{\Gamma \deltaz + (\Gamma - 1) b - 1}{b + \deltaz}$,
which is still positive when $ \deltaz \geq 2 $ for all $ \Gamma \geq 1 $.)

When testing $ \Hz : A_{0/1} \leq \az $, a bound on the worst-case p-value is
$ \Pr\left\{\Binom\left(a^\star + b^\star, \tfrac{\Gamma}{\Gamma + 1}\right) \geq a^\star \right\}$.
When testing $ \Hz^\prime : A_{0/1} \geq \az $, a bound on the
worst-case p-value is
$ \Pr\left\{\Binom\left(a^\star + b^\star, \tfrac{1}{\Gamma + 1}\right) \leq a^\star \right\}$.
Doubling the smaller of these bounds gives a bound on the two-sided
p-value. Expanded prediction sets consist of the points this procedure
does not reject. We may combine these prediction sets to form an
expanded confidence interval on the $\ATE$. The practical consequence:
we may compute a $\Gamma$-sensitivity confidence interval for the
$\ATE$ in $O(\log S)$ Binomial tail calculations.

\paragraph{Hodges-Lehmann point estimate under sensitivity.}
The Hodges-Lehmann strategy singles out the point equating the test
statistic to its null expectation, but under~(\ref{eqn:sensitivity}),
the test statistic has an unknown null expectation bounded between
$(1/(\Gamma + 1)) (a^\star + b^\star)$ and
$(\Gamma/(\Gamma + 1)) (a^\star + b^\star)$. The unknown
Hodges-Lehmann point estimate is therefore bounded between the roots
of $ a^\star = (1/(\Gamma + 1)) (a^\star + b^\star)$ and
$ a^\star = (\Gamma/(\Gamma + 1)) (a^\star + b^\star)$.

Consider first $ a^\star = (1/(\Gamma + 1)) (a^\star + b^\star)$, and
suppose
$ \Gamma \leq (\Szo + \Soo) / (\Szz + \Soz) $. We will show
$ \az = \Gamma \Soz - \Szo + (\Gamma - 1) \Szz $ is such a
root. First we show
$ a^\star = \min(\Szz + \Soz, \Soz + \Soo - \az) = \Szz + \Soz $.
We have
\begin{align}
  \Soz + \Soo - \az &= \Soo + \Szo - (\Gamma - 1) (\Szz + \Soz) \nonumber \\
     &\geq \Szz + \Soz, \nonumber
\end{align}
with the second line following from the bound on $ \Gamma $.
Next we show
$ b^\star = \min(\Szo + \Soo, \Szz + \Szo + \az) = \Szz + \Szo + \az $.
We have $ \Szz + \Szo + \az = \Gamma (\Szz + \Soz) \leq \Szo + \Soo$.
Finally,
\begin{align}
  \frac{1}{\Gamma + 1} (a^\star + b^\star) &= \frac{1}{\Gamma + 1}(\Szz + \Soz + \Szz + \Szo + \az) \nonumber \\
  &= \frac{1}{\Gamma + 1} [(\Gamma + 1) (\Szz + \Soz)] \nonumber \\
  &= \Szz + \Soz \nonumber \\
  &= a^\star. \nonumber
\end{align}

Now suppose $ \Gamma \geq (\Szo + \Soo) / (\Szz + \Soz) $. We will
show $ \az = \Soz - (1/\Gamma) \Szo + ((\Gamma - 1)/\Gamma) \Soo $
is the corresponding root. The argument parallels the previous. We have:
\begin{align}
  \Soz + \Soo - \az &= (\Szo + \Soo) / \Gamma \nonumber \\
   &\leq \Szz + \Soz, \nonumber
\end{align}
so $ a^\star = \Soz + \Soo - \az $. We have
\begin{align}
  \Szz + \Szo + \az &= \Szz + \Soz + \frac{\Gamma - 1}{\Gamma} (\Szo + \Soo) \nonumber \\
  &\geq \Szz + \Soz + \left(1 - \frac{\Szz + \Soz}{\Szo + \Soo} \right) (\Szo + \Soo) \nonumber \\
  &= \Szo + \Soo, \nonumber
\end{align}
so $ b^\star = \Szo + \Soo $. Finally,
\begin{align}
  \frac{1}{\Gamma + 1} (a^\star + b^\star)
     &= \frac{1}{\Gamma + 1} \left( \frac{1}{\Gamma}(\Szo + \Soo) + \Szo + \Soo \right) \nonumber \\
     &= \frac{1}{\Gamma + 1} \cdot \frac{\Gamma + 1}{\Gamma} (\Szo + \Soo) \nonumber \\
     &= \frac{1}{\Gamma}(\Szo + \Soo) \nonumber \\
     &= a^\star. \nonumber
\end{align}

Thus, one bound (the upper as we'll see) on the Hodges-Lehmann point estimate is:
\begin{displaymath}
  \mathrm{ub}_1 = \begin{cases}
    \Gamma \Soz - \Szo + (\Gamma - 1) \Szz & \Gamma \leq \frac{\Szo + \Soo}{\Szz + \Soz} \\
    \Soz - \frac{1}{\Gamma} \Szo + \frac{\Gamma - 1}{\Gamma} \Soo & \Gamma \geq \frac{\Szo + \Soo}{\Szz + \Soz}.
  \end{cases}
\end{displaymath}
As $\Gamma \to \infty$, $ \mathrm{ub}_1 \to \Soz + \Soo $,
the number of treated successes and the a priori upper bound on the
net effect.

Turning now to the other bound, we seek the root of
$ a^\star = (\Gamma/(\Gamma + 1)) (a^\star + b^\star)$.
Suppose
$ \Gamma \leq (\Szz + \Soz) / (\Szo + \Soo) $. We will show
$ \az = \Soz - \Gamma \Szo - (\Gamma - 1) \Soo $ is such a
root. First we show
$ a^\star = \min(\Szz + \Soz, \Soz + \Soo - \az) = \Soz + \Soo - \az $.
We have
\begin{align}
  \Soz + \Soo - \az &= \Gamma (\Szo + \Soo) \nonumber \\
     &\leq \Szz + \Soz, \nonumber
\end{align}
with the second line following from the bound on $ \Gamma $.
Next we show
$ b^\star = \min(\Szo + \Soo, \Szz + \Szo + \az) = \Szo + \Soo $.
We have
\begin{align}
  \Szz + \Szo + \az &= \Szz + \Soz - (\Gamma - 1)(\Szo + \Soo) \nonumber \\
     &\geq \Szo + \Soo, \nonumber
\end{align}
again from the bound on $ \Gamma $. Finally,
\begin{align}
  \frac{\Gamma}{\Gamma + 1} (a^\star + b^\star)
     &= \frac{\Gamma}{\Gamma + 1} (\Gamma (\Szo + \Soo) + \Szo + \Soo) \nonumber \\
     &= \frac{\Gamma}{\Gamma + 1} (\Gamma + 1)(\Szo + \Soo) \nonumber \\
     &= \Gamma (\Szo + \Soo) \nonumber \\
     &= a^\star. \nonumber
\end{align}

Now suppose $ \Gamma \geq (\Szz + \Soz) / (\Szo + \Soo) $. We will
show $ \az = (1/\Gamma) \Soz - \Szo - ((\Gamma - 1)/\Gamma) \Szz $
is the corresponding root. The argument parallels the previous. First we
show
$ a^\star = \min(\Szz + \Soz, \Soz + \Soo - \az) = \Szz + \Soz $.
We have
\begin{align}
  \Soz + \Soo - \az &= (\Szo + \Soo) + \frac{\Gamma - 1}{\Gamma} (\Szz + \Soz) \nonumber \\
     &\geq \Szz + \Soz, \nonumber
\end{align}
with the second line following from the bound on $ \Gamma $.
Next we show
$ b^\star = \min(\Szo + \Soo, \Szz + \Szo + \az) = \Szz + \Szo + \az $.
We have $ \Szz + \Szo + \az = \frac{1}{\Gamma} (\Szz + \Soz) \leq \Szo + \Soo$.
Finally,
\begin{align}
  \frac{\Gamma}{\Gamma + 1} (a^\star + b^\star)
     &= \frac{\Gamma}{\Gamma + 1} \left( \Szz + \Soz + \frac{1}{\Gamma}(\Szz + \Soz) \right) \nonumber \\
     &= \frac{\Gamma}{\Gamma + 1} \cdot \frac{\Gamma + 1}{\Gamma} (\Szz + \Soz) \nonumber \\
     &= \Szz + \Soz \nonumber \\
     &= a^\star. \nonumber
\end{align}

Thus, the lower bound on the Hodges-Lehmann point estimate is:
\begin{displaymath}
  \mathrm{lb}_1 = \begin{cases}
    \Soz - \Gamma \Szo - (\Gamma - 1) \Soo & \Gamma \leq \frac{\Szz + \Soz}{\Szo + \Soo} \\
    \frac{1}{\Gamma} \Soz - \Szo - \frac{\Gamma - 1}{\Gamma} \Szz & \Gamma \geq \frac{\Szz + \Soz}{\Szo + \Soo},
  \end{cases}
\end{displaymath}
which goes to $ -(\Szz + \Szo) $ as $ \Gamma \to \infty, $
the negative of the number of treated failures and the a priori lower
bound on the net effect.

The corresponding bounds for $\Az$ follow from the $\Ao$ derivation by
exchanging the roles of $\Szz$ and $\Soo$ throughout, leaving
$\hat A = \Soz - \Szo$ unchanged. Under the monotonicity assumption
the same root calculations go through with the reduced box ceilings of
\S\ref{sec:hardest} (for $\Ao$, the right-face ceiling
$\Szz + \Soz$ becomes $\Soz$; for $\Az$, $\Soz + \Soo$
becomes $\Soz$), with one new feature. Monotonicity forces
$\mpo = \Szo + \moo \geq \Szo$, so the worst-case $b^\star$
cannot fall below $\Szo$. The upper bounds are unaffected (they
still saturate at the a priori maxima $s_T = \Soz + \Soo$ and
$f_C = \Szz + \Soz$) but the lower bounds, driven toward small
$b^\star$ as $\Gamma$ grows, meet this floor: each saturates exactly at
the a priori minimum $0$ once $\Gamma \geq \Soz / \Szo$, rather than
continuing to the non-monotonic limits $-f_T$ and $-s_C$.
Table~\ref{tbl:sensitivity} collects all four estimand/assumption
combinations, writing $f_T, s_T$ ($f_C, s_C$) for the numbers of
treated (control) failures and successes.

Combining the non-monotonic $\Ao$ and $\Az$ bounds of
Table~\ref{tbl:sensitivity} through the sum
constraint~(\ref{eqn:sum-constraint}), the $\ATE$ point estimate under
sensitivity lies in
\begin{displaymath}
  \left[ \frac{\mathrm{lb}_1 + \mathrm{lb}_0}{2S}, \;
         \frac{\mathrm{ub}_1 + \mathrm{ub}_0}{2S} \right],
\end{displaymath}
which widens monotonically with $\Gamma$ and saturates as
$\Gamma \to \infty$ at
\begin{displaymath}
  \frac{[-(\Szz + 2 \Szo + \Soo),\; \Szz + 2 \Soz + \Soo]}{2S}.
\end{displaymath}
We may write these endpoints as
$ \tfrac{\hat{A}}{2S} \pm \tfrac{1}{2} $. This is the partial
identification interval of \citet{manski-2003-partial-identification}.
The $\Gamma$-model thus interpolates continuously between the
randomization interval at $\Gamma = 1$ and the assumption-free Manski
region as $\Gamma \to \infty$.

\begin{table}[p]
\centering
\footnotesize
\setlength{\tabcolsep}{5pt}
\caption{Bounds on the Hodges-Lehmann point estimate under the
$\Gamma$-sensitivity model, by estimand and monotonicity assumption.
Each cell gives the bound as a function of $\Gamma$ (a two-case
closed form) together with its $\Gamma \to \infty$ limit. The $\Az$
bounds follow from the $\Ao$ derivation by exchanging $\Szz$ and
$\Soo$. The monotonic bounds use the reduced box ceilings of
\S\ref{sec:hardest}. Here $f_T = \Szz + \Szo$ and
$s_T = \Soz + \Soo$ are the numbers of treated failures and
successes, $f_C = \Szz + \Soz$ and $s_C = \Szo + \Soo$ the
control counts. Under monotonicity the lower bounds saturate exactly
at the a priori minimum $0$ once $\Gamma \geq \Soz / \Szo$.}
\label{tbl:sensitivity}
\begin{tabular}{c c l l}
\toprule
Estimand & \shortstack{Monotonic\\ Treatment\\ Effects?} & Lower Bound & Upper Bound \\
\midrule
\multirow{2}{*}{$\Ao$} & N
 & $\begin{array}{@{}l@{}}
     \begin{cases}
       \Soz - \Gamma \Szo - (\Gamma-1)\Soo & \Gamma \leq \tfrac{\Szz+\Soz}{\Szo+\Soo} \\[3pt]
       \tfrac{1}{\Gamma}\Soz - \Szo - \tfrac{\Gamma-1}{\Gamma}\Szz & \Gamma \geq \tfrac{\Szz+\Soz}{\Szo+\Soo}
     \end{cases} \\[6pt]
     \quad \to\, -f_T
   \end{array}$
 & $\begin{array}{@{}l@{}}
     \begin{cases}
       \Gamma \Soz - \Szo + (\Gamma-1)\Szz & \Gamma \leq \tfrac{\Szo+\Soo}{\Szz+\Soz} \\[3pt]
       \Soz - \tfrac{1}{\Gamma}\Szo + \tfrac{\Gamma-1}{\Gamma}\Soo & \Gamma \geq \tfrac{\Szo+\Soo}{\Szz+\Soz}
     \end{cases} \\[6pt]
     \quad \to\, s_T
   \end{array}$ \\
\cmidrule(l){2-4}
 & Y
 & $\begin{array}{@{}l@{}}
     \begin{cases}
       \Soz - \Gamma \Szo - (\Gamma-1)\Soo & \Gamma \leq \tfrac{\Soz}{\Szo+\Soo} \\[3pt]
       \tfrac{1}{\Gamma}\Soz - \Szo & \tfrac{\Soz}{\Szo+\Soo} \leq \Gamma \leq \tfrac{\Soz}{\Szo} \\[3pt]
       0 & \Gamma \geq \tfrac{\Soz}{\Szo}
     \end{cases}
   \end{array}$
 & $\begin{array}{@{}l@{}}
     \begin{cases}
       \Gamma \Soz - \Szo & \Gamma \leq \tfrac{\Szo+\Soo}{\Soz} \\[3pt]
       \Soz - \tfrac{1}{\Gamma}\Szo + \tfrac{\Gamma-1}{\Gamma}\Soo & \Gamma \geq \tfrac{\Szo+\Soo}{\Soz}
     \end{cases} \\[6pt]
     \quad \to\, s_T
   \end{array}$ \\
\midrule
\multirow{2}{*}{$\Az$} & N
 & $\begin{array}{@{}l@{}}
     \begin{cases}
       \Soz - \Gamma \Szo - (\Gamma-1)\Szz & \Gamma \leq \tfrac{\Soz+\Soo}{\Szz+\Szo} \\[3pt]
       \tfrac{1}{\Gamma}\Soz - \Szo - \tfrac{\Gamma-1}{\Gamma}\Soo & \Gamma \geq \tfrac{\Soz+\Soo}{\Szz+\Szo}
     \end{cases} \\[6pt]
     \quad \to\, -s_C
   \end{array}$
 & $\begin{array}{@{}l@{}}
     \begin{cases}
       \Gamma \Soz - \Szo + (\Gamma-1)\Soo & \Gamma \leq \tfrac{\Szz+\Szo}{\Soz+\Soo} \\[3pt]
       \Soz - \tfrac{1}{\Gamma}\Szo + \tfrac{\Gamma-1}{\Gamma}\Szz & \Gamma \geq \tfrac{\Szz+\Szo}{\Soz+\Soo}
     \end{cases} \\[6pt]
     \quad \to\, f_C
   \end{array}$ \\
\cmidrule(l){2-4}
 & Y
 & $\begin{array}{@{}l@{}}
     \begin{cases}
       \Soz - \Gamma \Szo - (\Gamma-1)\Szz & \Gamma \leq \tfrac{\Soz}{\Szz+\Szo} \\[3pt]
       \tfrac{1}{\Gamma}\Soz - \Szo & \tfrac{\Soz}{\Szz+\Szo} \leq \Gamma \leq \tfrac{\Soz}{\Szo} \\[3pt]
       0 & \Gamma \geq \tfrac{\Soz}{\Szo}
     \end{cases}
   \end{array}$
 & $\begin{array}{@{}l@{}}
     \begin{cases}
       \Gamma \Soz - \Szo & \Gamma \leq \tfrac{\Szz+\Szo}{\Soz} \\[3pt]
       \Soz - \tfrac{1}{\Gamma}\Szo + \tfrac{\Gamma-1}{\Gamma}\Szz & \Gamma \geq \tfrac{\Szz+\Szo}{\Soz}
     \end{cases} \\[6pt]
     \quad \to\, f_C
   \end{array}$ \\
\bottomrule
\end{tabular}
\end{table}

\paragraph{Large-sample sensitivity intervals.}
The Gaussian approximation of \S\ref{sec:practical} extends to the
sensitivity model with a single change. The worst-case corner
$(a^\star, b^\star)$ is exactly that of \S\ref{sec:hardest}:
only the reference distribution moves, from
$\Binom(a^\star + b^\star, \tfrac{1}{2})$ to
$\Binom(a^\star + b^\star, \pi)$ with $\pi = \Gamma/(\Gamma+1)$ for the
right-tailed (lower-endpoint) test and $\pi = 1/(\Gamma+1)$ for the
left-tailed. Standardizing the worst-case statistic against its
shifted null mean $\pi n$ and variance $\pi(1-\pi) n$, where
$n := a^\star + b^\star$, with the continuity correction
$a^\star \mapsto a^\star - \tfrac{1}{2}$ (right tail),
\begin{equation}\label{eqn:Z-sensitivity}
   Z = \frac{a^\star - \tfrac{1}{2} - \pi n}{\sqrt{\pi(1-\pi)\, n}}
     = \frac{(\Gamma+1)(\deltaz - 1) - (\Gamma-1)\, n}{2\sqrt{\Gamma\, n}},
\end{equation}
which recovers the standardization of \S\ref{sec:practical},
$(\deltaz - 1)/\sqrt{2 b^\star + \deltaz}$, at $\Gamma = 1$.

Setting $Z = \zalpha$ and substituting the face-dependent count
($n = 2(\Szo + \Soo) + \deltaz$ where the consistency line
meets the top face, $n = 2(\Szz + \Soz) - \deltaz$ where it meets
the right face) gives in each case a quadratic in $\sqrt{n}$ whose
positive root is elementary. Writing $z := \zalpha$ and collecting
terms, every endpoint is $\hat A$ plus or minus one of two kernels,
\begin{align}
   M(C) &:= (\Gamma-1) C + \frac{\Gamma+1}{2} + \frac{\Gamma}{2} z^2
          + z \sqrt{\Gamma(\Gamma+1)\bigl(C + \tfrac{1}{2}\bigr)
                    + \tfrac{1}{4} \Gamma^2 z^2}, \label{eqn:kernel-M} \\
   R(C) &:= \frac{\Gamma-1}{\Gamma} C + \frac{\Gamma+1}{2\Gamma}
          - \frac{z^2}{2\Gamma}
          + \frac{z}{\Gamma} \sqrt{(\Gamma+1)\bigl(C - \tfrac{1}{2}\bigr)
                    + \tfrac{1}{4} z^2}, \label{eqn:kernel-R}
\end{align}
evaluated at the relevant box ceiling $C$. The lower endpoint is
$\Lalpha = \hat A - M(c)$ when the consistency line meets the top face
and $\hat A - R(d)$ when it meets the right face; the upper endpoint is
$\Ualpha = \hat A + R(c)$ (top face) or $\hat A + M(d)$ (right face).
Here $c$ and $d$ are the top- and right-face ceilings of
\S\ref{sec:hardest}: $c = \Szo + \Soo$ and $d = \Szz + \Soz$
for $\Ao$, the $\Szz \leftrightarrow \Soo$ swap for $\Az$, and the
reduced right ceiling $d = \Soz$ under monotonicity.
Table~\ref{tbl:sensitivity-large-sample} collects all sixteen
endpoints; as in \S\ref{sec:practical} the applicable face is the one
whose endpoint satisfies the stated criterion, and the integer
endpoints are $\Lalpha = \lceil \cdot \rceil$ and
$\Ualpha = \lfloor \cdot \rfloor$.

At $\Gamma = 1$ the two kernels coincide with the continuity-corrected
roots of \S\ref{sec:practical}, and
Table~\ref{tbl:sensitivity-large-sample} reduces to
Table~\ref{tbl:large-sample}. Dropping the continuity correction and
setting $z = 0$ makes $Z = 0$, which is the fixed point
$a^\star = \pi(a^\star + b^\star)$ solved above, so the interval
collapses to the Hodges-Lehmann sensitivity bounds of
Table~\ref{tbl:sensitivity}. The continuity correction now scales as
$(\Gamma+1)/2$ rather than a fixed unit, widening the interval by more
than one count at each end as $\Gamma$ grows. And because the top- and
right-face kernels carry factors of $\Gamma$ and $1/\Gamma$
respectively, the interval is no longer symmetric about $\hat A$:
sensitivity drifts its center along the point-estimate range of
Table~\ref{tbl:sensitivity} while the $z$-terms supply the sampling
half-width, so there is no clean $\hat A \pm$ half-width form like
(\ref{eqn:LU-approx}). Under monotonicity the lower endpoints
additionally saturate at the a-priori minimum $0$ once
$\Gamma \geq \Soz/\Szo$, exactly as the point-estimate bounds do
(Table~\ref{tbl:sensitivity}).

\begin{table}[p]
\centering
\footnotesize
\setlength{\tabcolsep}{6pt}
\caption{Large-sample Gaussian approximations to the prediction-set
endpoints under the $\Gamma$-sensitivity model, by estimand,
monotonicity assumption, endpoint (lower $\Lalpha$ or upper
$\Ualpha$), and which face of the constraint box the consistency line
meets: top ($T$) or right ($R$). Each entry uses the kernels
$M(\cdot)$ and $R(\cdot)$ of (\ref{eqn:kernel-M})--(\ref{eqn:kernel-R})
with $z := \zalpha$; the applicable face is the one whose endpoint
satisfies the criterion in the final column, and the integer endpoints
are $\Lalpha = \lceil \cdot \rceil$, $\Ualpha = \lfloor \cdot
\rfloor$. Setting $\Gamma = 1$ recovers Table~\ref{tbl:large-sample}.
Under monotonicity the lower endpoints saturate at the a-priori minimum
$0$ once $\Gamma \geq \Soz/\Szo$.}
\label{tbl:sensitivity-large-sample}
\begin{tabular}{c c l c}
\toprule
Estimand & \shortstack{Monotonic\\ Treatment\\ Effects?} & Formula & Criterion \\
\midrule
\multirow{8}{*}{$\Ao$}
 & \multirow{4}{*}{N}
   & $\Lalpha^T = \hat A - M(\Szo+\Soo)$ & $\Lalpha^T \geq \Soo-\Szz$ \\
 & & $\Lalpha^R = \hat A - R(\Szz+\Soz)$ & $\Lalpha^R < \Soo-\Szz$ \\
 & & $\Ualpha^T = \hat A + R(\Szo+\Soo)$ & $\Ualpha^T \geq \Soo-\Szz$ \\
 & & $\Ualpha^R = \hat A + M(\Szz+\Soz)$ & $\Ualpha^R < \Soo-\Szz$ \\
\cmidrule(l){2-4}
 & \multirow{4}{*}{Y}
   & $\Lalpha^T = \hat A - M(\Szo+\Soo)$ & $\Lalpha^T \geq \Soo$ \\
 & & $\Lalpha^R = \hat A - R(\Soz)$ & $\Lalpha^R < \Soo$ \\
 & & $\Ualpha^T = \hat A + R(\Szo+\Soo)$ & $\Ualpha^T \geq \Soo$ \\
 & & $\Ualpha^R = \hat A + M(\Soz)$ & $\Ualpha^R < \Soo$ \\
\midrule
\multirow{8}{*}{$\Az$}
 & \multirow{4}{*}{N}
   & $\Lalpha^T = \hat A - M(\Szz+\Szo)$ & $\Lalpha^T \geq \Szz-\Soo$ \\
 & & $\Lalpha^R = \hat A - R(\Soz+\Soo)$ & $\Lalpha^R < \Szz-\Soo$ \\
 & & $\Ualpha^T = \hat A + R(\Szz+\Szo)$ & $\Ualpha^T \geq \Szz-\Soo$ \\
 & & $\Ualpha^R = \hat A + M(\Soz+\Soo)$ & $\Ualpha^R < \Szz-\Soo$ \\
\cmidrule(l){2-4}
 & \multirow{4}{*}{Y}
   & $\Lalpha^T = \hat A - M(\Szz+\Szo)$ & $\Lalpha^T \geq \Szz$ \\
 & & $\Lalpha^R = \hat A - R(\Soz)$ & $\Lalpha^R < \Szz$ \\
 & & $\Ualpha^T = \hat A + R(\Szz+\Szo)$ & $\Ualpha^T \geq \Szz$ \\
 & & $\Ualpha^R = \hat A + M(\Soz)$ & $\Ualpha^R < \Szz$ \\
\bottomrule
\end{tabular}
\end{table}

Combining the $\Ao$ and $\Az$ endpoints at per-side level $\alpha/4$
yields the large-sample $\Gamma$-sensitivity confidence interval for
the $\ATE$,
\begin{displaymath}
   \left[ \frac{L_1 + L_0}{2S}, \; \frac{U_1 + U_0}{2S} \right].
\end{displaymath}
Read $(L_1, U_1)$ and $(L_0, U_0)$ from
Table~\ref{tbl:sensitivity-large-sample} for $\Ao$ and $\Az$ at
$z = \zalphaq$. For the running example at
$\Gamma = 1$ this reproduces the exact interval of \S\ref{sec:ate},
$[0.0055, 0.0745]$; at $\Gamma = 1.25$ it widens and drifts to
$[-0.029, 0.111]$, the asymmetry reflecting the sensitivity-induced
shift of the center.

\paragraph{The sensitivity value.}
While the amount of hidden bias is typically unknown, we may calculate
the point at which we begin to question a study's conclusions. This
occurs at the smallest $\Gamma$ at which the study fails to reject the
null of no positive net effect at level $\alpha$. Call this threshold
$ \Gamma^\bullet $ the \emph{sensitivity value}
\citep{zhao-2019-sensitivity}. Debates over study conclusions can then
focus on whether the amount of hidden bias is below or above this
threshold, how we might detect such biases, and how we might design
studies to have correspondingly larger thresholds. A finding with
$\Gamma^\bullet$ well above $1$ is robust to substantial unmeasured
confounding; $\Gamma^\bullet$ close to $ 1 $ would be cause for
concern in an observational study, because rarely could we be so
confident in our adjustment strategy.

For the running example, testing $\Hz: \ATE \leq 0$ at
$\alpha = 0.10$, the sensitivity value is
$\Gamma^\bullet \approx 1.03$. The positive net effect is thus
sensitive to small hidden biases: an unmeasured confounder shifting
the within-pair odds of treatment by as little as 3--4\% could explain
the effect. Figure~\ref{fig:sensitivity} traces the two intervals of
this section as $\Gamma$ grows: the sensitivity interval quantifying
uncertainty from hidden biases, nested within the expanded confidence
interval that also carries stochastic assignment uncertainty. Both
intervals straddle the Hodges-Lehmann point estimate
$\widehat{\ATE} = 0.04$, with the expanded interval crossing zero at
$\Gamma^\bullet$.

\begin{figure}[t]
\centering
\includegraphics[width=0.82\linewidth]{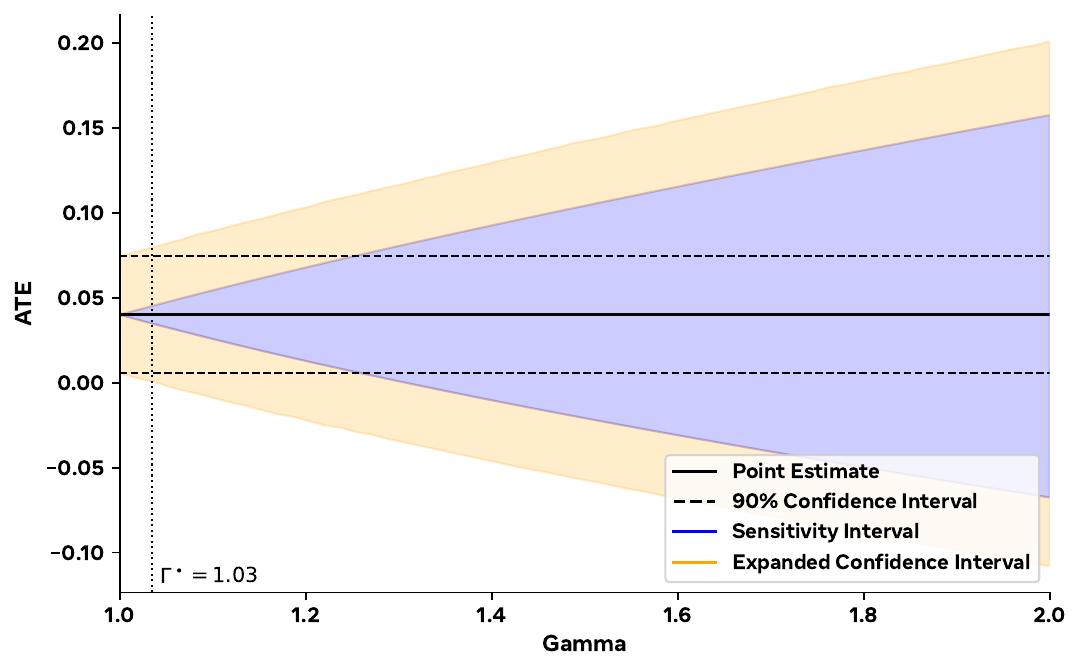}
\caption{Sensitivity of the $\ATE$ for the running example
($\Szz = 800$, $\Szo = 30$, $\Soz = 70$, $\Soo = 100$). As the
hidden-bias odds ratio $\Gamma$ grows from $1$ (randomization), two
intervals widen around the Hodges-Lehmann point estimate
$\widehat{\ATE} = 0.04$ (solid black): the confounding-only
\emph{sensitivity interval} (blue), and the \emph{expanded confidence
interval} (orange) that also carries stochastic uncertainty at
coverage $1 - \alpha = 90\%$. Dashed lines mark the $\Gamma = 1$
confidence interval $[0.0055, 0.0745]$; the dotted vertical line marks
the sensitivity value $\Gamma^\bullet \approx 1.03$, where the
expanded interval first reaches the null $\ATE = 0$.}
\label{fig:sensitivity}
\end{figure}

\paragraph{Design sensitivity.}
The value $\Gamma^\bullet$ mixes two ingredients: sensitivity to bias
and ordinary stochastic noise from the probabilistic assignment
mechanism. A simple argument finds a value of $ \Gamma $ preventing us
from rejecting the hypothesis of 0 net effect at any sample size.

Recall $ a^\star \sim \Binom(a^\star + b^\star, \Gamma / (\Gamma + 1)) $
under $ \Hz $. Suppose $ \Gamma $ is such that $ a^\star $ equalled
its null expectation:
\begin{align}
  a^\star &= \frac{\Gamma}{\Gamma + 1} (a^\star + b^\star) \nonumber \\
  \Rightarrow \Gamma &= \frac{a^\star}{b^\star} \nonumber \\
  &= \frac{ \Szo + \Soo + \deltaz }{\Szo +\Soo} \nonumber \\
  &= \frac{ \Soz + \Soo }{\Szo +\Soo} \nonumber \\
  &= \frac{ p_{1 +} }{ p_{+ 1} } = 1 + \frac{\tau}{p_{+1}} =: \tilde{\Gamma}, \nonumber
\end{align}
where the third line substitutes Equation~(\ref{eqn:worst}), without
assuming monotonicity and assuming the diagonal consistency line
intersects the top of the box. (As discussed in \S\ref{sec:hardest}, this
is the case whenever $ \Soo \leq \Szz $ with $ \az = 0 $, that is,
when failures are more common than successes.) The last line expresses
the threshold as the ratio of the treated success rate to the control
success rate, and as $ 1 $ plus the relative treatment effect compared
to the baseline success rate, $ \tau / p_{+1}, $ where $ \tau = p_{1+} - p_{+1}. $

If the unknown $ \Gamma $ does exceed that threshold, we would not be
able to reject the null hypothesis of 0 net effect at any sample size.
Thus, the \emph{design sensitivity} $ \tilde{\Gamma} $ offers a
valuable yardstick for assessing our hope of detecting an effect. In
an experiment, it is recommended to perform \emph{power analysis} as
part of the design. Power analysis requires we reflect on the success
rates we expect, often based on previous analyses, and the effect
size we hope to detect, telling us the sample size we need. Design
sensitivity serves a similar role in observational studies, telling us
how thorough our adjustment strategy needs to be.

Based on a relative effect size we hope to detect, we must apply our
domain expertise and ensure we measure the most important confounders.
If we lack such ability, there is little hope of success. We may
employ quasi-experimental devices to rule out the possibility of large
hidden biases \citep{rosenbaum-1989-known-effects, rosenbaum-2015-quasi-experiment}.

This discussion demonstrates the hopelessness of measuring small
relative effect sizes in an observational study. No matter our domain
expertise, rarely could we be confident the remaining ``hidden''
biases were minuscule. Thus there is no hope of measuring minuscule
relative treatment effects with confidence in an observational study.

In the running example, the observed $ p_{+ 1} = 0.13 $, and the
observed treatment effect is $ +0.04 $ for a \emph{post-hoc} design
sensitivity of $ 1.31. $ The study performed in this section did not
yield an impressive sensitivity value, but repeating it with a larger
sample size would not help much. Design sensitivity is most useful
\emph{before} we have performed a study; warnings analogous to those
against post-hoc power analysis apply.

When we assume monotonic treatment effects,
$ \tilde{\Gamma} = \Soz / \Szo = 1 + \hat{A} / \Szo = 1 + \tau / p_{01}. $
The smaller denominator of the monotonic case (compared to the
non-monotonic case, $ p_{+ 1} = p_{01} + p_{11} $) aids the design
sensitivity, but the formula lacks the clean interpretation of the
non-monotonic case.

\section{Discussion}\label{sec:discussion}
Textbook discussions of paired studies with binary outcomes often take
a population sampling approach \citep{fleiss-2013-proportions}.
When inference is instead rooted in random assignment, it is often
limited to testing the null hypothesis of no treatment effect for any
unit \citep{lehmann-1975-nonparametrics}. This note has demonstrated
how the randomized treatment assignment provides the ``reasoned basis
for inference'' for effect sizes. This leads to different formulae
than sampling-based derivations and is more defensible when units are
\emph{not} samples from a population, when interest lies in the
effect of treatment on the units in the study, or in an observational
study where hidden biases from unobserved factors may be present.

Previous work assumed monotonicity
\citep{rosenbaum-2002-attributing-effects-to-treatment-in-matched-studies}
or relied on an integer program solver
\citep{rigdon-2015-attributable-effect}. The present note provides
analytical formulae for point estimates, p-values (against arbitrary
null hypotheses), and large-sample approximations for confidence
bounds (exact confidence bounds in $O(\log S)$) based on inverting
tests. It also provides formulae for the interval summarizing
sensitivity of point estimates to hidden biases and expanded
confidence intervals for observational studies.

When we know treatment effects are monotonic, incorporating this
assumption achieves greater statistical power and narrower confidence
intervals. In many applications, this assumption is not defensible,
and our discussion shows it is unnecessary for inference.

\bibliographystyle{plainnat}
\bibliography{citations}

\end{document}